\documentclass[acmsmall,nonacm]{acmart}
\usepackage{amsmath}
\usepackage{amsthm}
\usepackage{amssymb}
\usepackage{wrapfig}
\usepackage[capitalize]{cleveref}
\usepackage{graphicx}
\usepackage{subcaption}
\usepackage{todonotes}
\presetkeys{todonotes}{inline}{}
\usepackage{enumerate}
\usepackage{bm}

\usepackage{booktabs}

\makeatletter
\newtheorem*{rep@theorem}{\rep@title}
\newcommand{\newreptheorem}[2]{%
\newenvironment{rep#1}[1]{%
 \def\rep@title{#2 \ref{##1}}%
 \begin{rep@theorem}}%
 {\end{rep@theorem}}}
\makeatother

\newtheorem{claim}{Claim}

\newreptheorem{lemma}{Lemma}
\newreptheorem{theorem}{Theorem}
\newreptheorem{claim}{Claim}

\newcommand{\Exp}{\mathrm{Exp}}
\newcommand{\idle}{\mathrm{idle}}
\newcommand{\busy}{\mathrm{busy}}
\newcommand{\naive}{\mathrm{naive}}

\newcommand{\p}{\mathfrak{p}}

\title{Stability for Two-class Multiserver-job Systems}
\author{Isaac Grosof}
\affiliation{%
    \institution{Carnegie Mellon University}
    \department{Computer Science Department}
    \city{Pittsburgh}
    \state{PA}
    \country{USA}}
\email{igrosof@cs.cmu.edu}
\author{Mor Harchol-Balter}
\affiliation{%
    \institution{Carnegie Mellon University}
    \department{Computer Science Department}
    \city{Pittsburgh}
    \state{PA}
    \country{USA}}
\email{harchol@cs.cmu.edu}
\author{Alan Scheller-Wolf}
\affiliation{%
    \institution{Carnegie Mellon University}
    \department{Tepper School of Business}
    \city{Pittsburgh}
    \state{PA}
    \country{USA}}
\email{awolf@andrew.cmu.edu}
\thanks{This work was supported by NSF-CMMI-1938909, NSF-XPS-1629444, and NSF-CSR-1763701 and by a Google 2020 Faculty Research Award.}
\begin{document}
\begin{abstract}
    Multiserver-job systems,
    where jobs require concurrent service at many servers,
    occur widely in practice.
    Much is known in the dropping setting,
    where jobs are immediately discarded if they require more servers than are currently available.
    However, very little is known in the more practical setting
    where jobs queue instead.

    In this paper, we derive a closed-form analytical expression for the stability region
    of a two-class (non-dropping) multiserver-job system
    where each class of jobs
    requires a distinct number of servers and requires a
    distinct exponential distribution of service time,
    and jobs are served in first-come-first-served (FCFS) order.
    This is the first result of any kind for an FCFS multiserver-job system
    where the classes have distinct service distributions.
    Our work is based on a technique
    that leverages the idea of a ``saturated'' system,
    in which an unlimited number of jobs are always available.

    Our analytical formula provides insight into the behavior of FCFS multiserver-job systems,
    highlighting the huge wastage (idle servers while jobs are in the queue)
    that can occur,
    as well as the nonmonotonic effects of the service rates on wastage.
\end{abstract}
\maketitle
\section{Introduction}

Traditional queueing theory is built on models, such as the M/G/k, where
every job occupies exactly one server,
however many servers are available.
These models have been popular for decades
because they capture the behavior of previous computing systems,
while admitting theoretical analysis.
However, traditional one-server-per-job models
are no longer representative of many modern computing systems.

Consider large-scale computing centers today,
such as those of Google, Facebook, and Microsoft.
Even though the \emph{servers} in these data centers still resemble
the \emph{servers} in traditional models such as the M/G/k,
the \emph{jobs} have changed:
these systems now by default have
jobs that require multiple servers.
For instance, in \cref{fig:cpu_requests},
we show the distribution of the number of CPUs requested
by jobs in Google's recently published trace of its ``Borg'' computation cluster
\cite{tirmazi_borg}.
The distribution is highly variable,
with jobs requesting anywhere from 1 to 100,000 normalized CPUs%
\footnote{The data was published in a scaled form \cite{tirmazi_borg}.
We rescale the data so the
smallest job in the trace uses one normalized CPU.}.
Throughout this paper, we will focus on what we call the ``multiserver-job model,''
by which we refer
to the common situation in modern systems where each job
occupies a fixed number of servers (typically more than one),
throughout its time in the system.

The multiserver-job model is fundamentally different from the one-server-per-job model.
For example, in the one-server-per-job model a work-conservation property holds,
where as long as enough jobs are present, no servers will be idle.
In the multiserver-job model work conservation is no longer guaranteed,
since a job might be forced to wait simply because
it demands more servers than are currently available,
and thus cannot ``fit,''
even though some servers are idle.
As a result, server utilization and system stability
are affected by the scheduling policy in the multiserver-job model,
unlike in a work conserving one-server-per-job model.
The multiserver-job state space is also much more complex,
rendering analysis far more difficult.

\subsection{Prior multiserver-job models}

Almost all existing work on multiserver-job systems has focused on
the dropping model,
where jobs that cannot receive service are dropped.
In a paper from 1979, \citet{arthurs_sizing}
consider this dropping model
and derive general analytical results
describing the steady state distribution.
In \cref{sec:dropping} we describe a few generalizations of
\cite{arthurs_sizing},
still within the context of the dropping model.

Unfortunately, the dropping model is unrealistic.
Large-scale systems run by companies like
Google, Facebook and Microsoft
have long queues to avoid dropping jobs,
as can be seen in Google's Borg trace \cite{tirmazi_borg}. 
Consequently,
we choose to study a
multiserver-job model which
assumes unbounded queues with no dropping.
We further assume that jobs that queue up
are served
in first-come-first-served (FCFS) order,
which is often the default used in production systems
\cite{etsion_short,sliwko_taxonomy}.

A few papers address
an FCFS multiserver-job model with more than two servers
in an analytic (non-numerical) manner
\cite{rumyantsev_stability,morozov_stability,afanaseva_stability}.
Each of these papers assumes that
all jobs have service times drawn from a
\emph{single} distribution,
regardless of the number of servers required by the job.
Having done so,
the papers derive analytical formulas
for their systems' stability regions.
Unfortunately, in real systems, jobs requiring different numbers of servers
typically also require different amounts of service time.
For instance, in \cref{fig:intro_correlation},
we show that there is a correlation between a job's number of requested CPUs
and the job's duration
for jobs
in Google's recent Borg trace \cite{tirmazi_borg}.

\begin{figure}
    \centering
    \begin{tabular}{cc}
        \cref{fig:cpu_requests} & \cref{fig:intro_correlation}\\
        \includegraphics[width=0.45\textwidth]{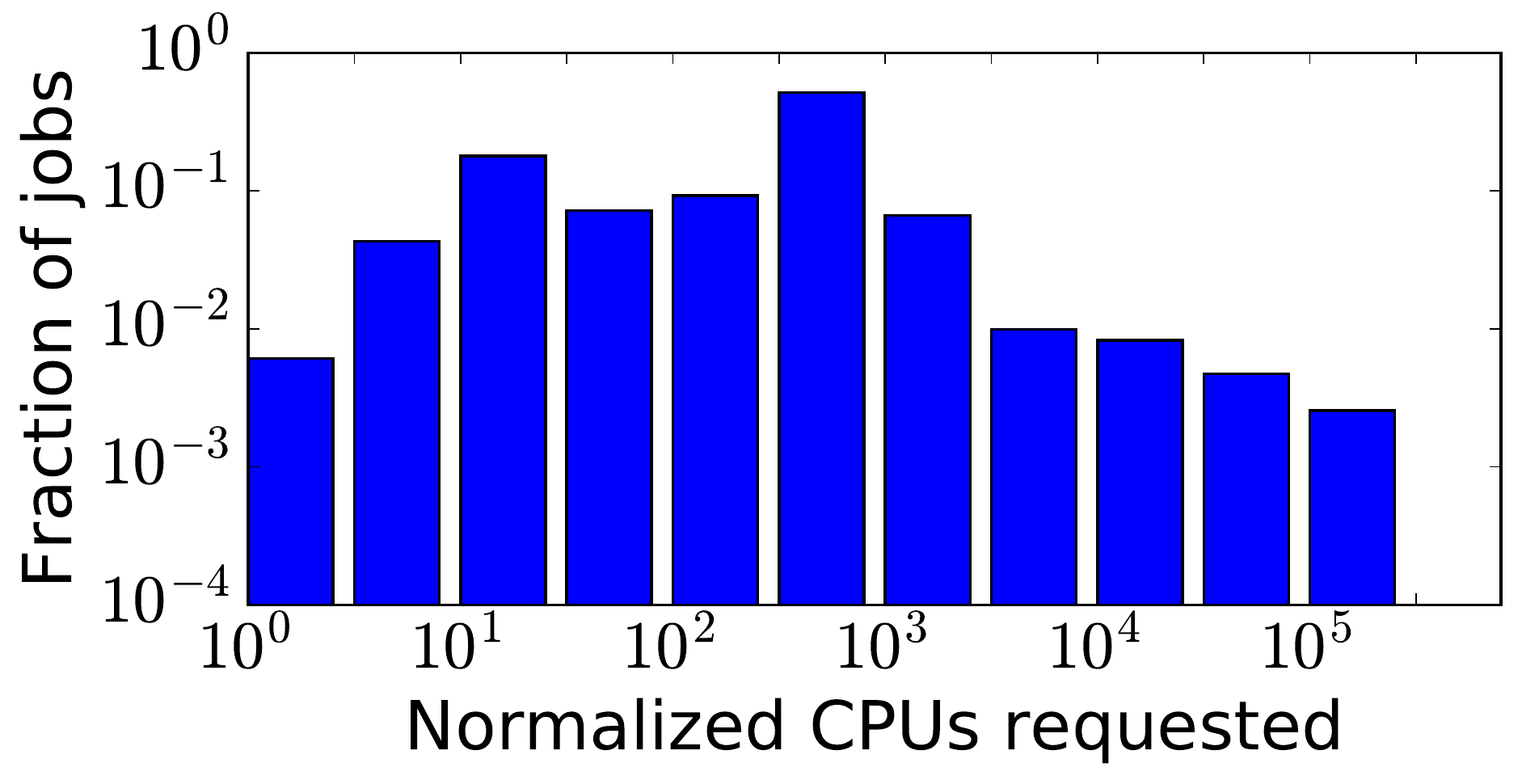} &
        \includegraphics[width=0.3\textwidth]{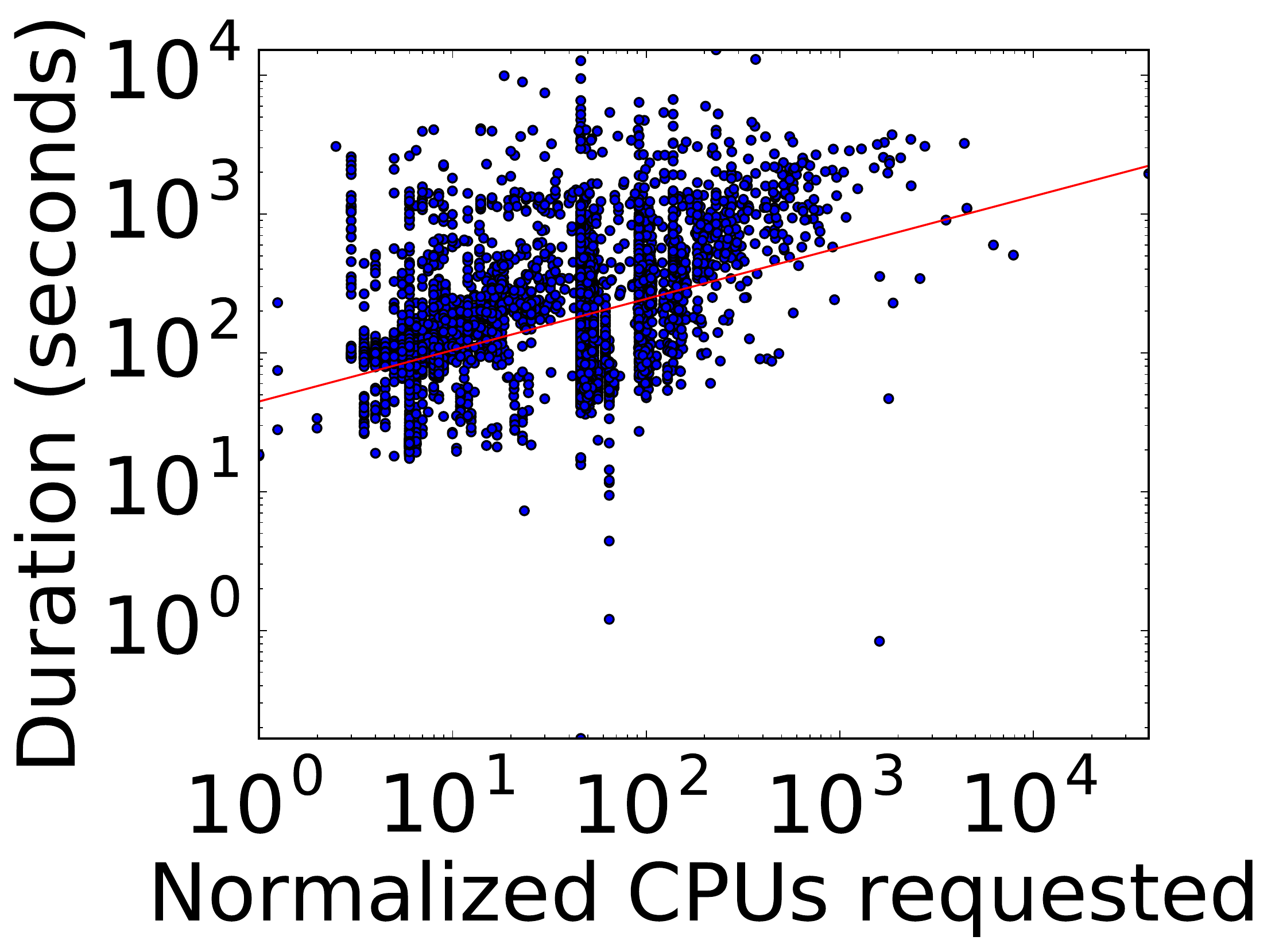}
    \end{tabular}
    \vspace{-12pt}
    \caption{The distribution of number of CPUs requested
    in Google's recently published Borg trace \cite{tirmazi_borg}.
    Number of CPUs is normalized to the size of the smallest request observed,
    not an absolute value.
    The peak of the distribution is around 500 normalized CPUs,
    and there are sizeable tails to either side.
    }
    \vspace{-8pt}
    \label{fig:cpu_requests}
    \caption{Correlation between CPUs requested and duration
    in Google's recently published Borg trace \cite{tirmazi_borg}.
    Here we display jobs in the Google ``free'' tier,
    where there are no latency guarantees.
    Number of CPUs is normalized to the size of the smallest request observed,
    not an absolute value.
    The red line shows a best-fit power-law approximation.
    The number of servers required by a job
    and its duration are clearly correlated,
    with a Pearson correlation coefficient (log-log) of 0.45.
    }
    \label{fig:intro_correlation}
\end{figure}

Thus, to handle real-world settings,
it is vital that our multiserver-job model should both allow
jobs to queue,
as well as allow
multiple classes of jobs with different service rates.
Unfortunately, when the classes have different service rates,
prior analytical techniques become inapplicable.
In this paper,
we take the first step in solving a multiclass multiserver-job model
with both FCFS queueing and different per-class service rates.
Due to the added complexity of having different service rates,
our analysis is limited to a two-class model.

\subsection{Our multiserver-job model}
\label{sec:intro_model}

The classes in our two-class multiserver-job model 
are labeled class 1 and class 2.
Class $i$ jobs have duration (size) distributed $\Exp(\mu_i)$,
and require a fixed number of servers, $n_i$, where $n_1 < n_2$.
The total number of servers available is $n$, where $n_2 \le n$.
We make no assumptions on the relationship between $\mu_1$ and $\mu_2$.
Jobs arrive to the system according to a Poisson process with rate $\lambda$.
Arriving jobs are independently in class $i$ with probability $p_i$.
Jobs that cannot immediately receive service queue up
and are served
in first-come-first-served (FCFS) order.
\subsection{Wastage and Stability}
\label{sec:intro_wastage}
In this paper we study the stability region of the model given
in \cref{sec:intro_model}.
We derive the maximum arrival rate $\lambda^*$
such that the system is stable
(positive recurrent)
for any arrival rate $\lambda < \lambda^*$.

The key to understanding stability in the multiserver-job system
is understanding ``wasted servers'' or ``wastage,''
which we can think of as the number of servers which are idle while at least one job is in the queue.
Understanding stability and wastage in a multiclass multiserver-job model
is a difficult open problem, and is fundamental to capacity provisioning for today's
data centers.

To make the idea of wastage concrete,
first let us define
$N_\idle$ to be the number of servers which are \textit{idle} in steady state,
and define $N_\busy$ similarly.
Note that $N_\idle + N_\busy = n$.

Let us define the number of wasted servers, or the ``wastage,''
to be $E[N_\idle | \text{queue nonempty}]$.
For understanding stability,
the most important aspect of wastage is the ``limiting wastage,''
which we write as $E[N_\idle^*]$ ($E[N_\busy^*]$ is defined analogously):
\[ E[N_\idle^*]
= \lim_{\lambda \to \lambda^*} E[N_\idle]
= \lim_{\lambda \to \lambda^*} E[N_\idle | \text{queue nonempty}]. \]
The second equality holds because the probability
that the queue is nonempty goes to 1 as $\lambda \to \lambda^*$.
When it is clear from context,
we will sometimes refer to ``limiting wastage'' as simply ``wastage.''

To relate the number of wasted servers to the stability region,
let us define $S$ to be the distribution of server-seconds demanded per job, i.e.
the number of servers demanded multiplied by the time demanded.
At the border of stability, $\lambda^*$ jobs per second are arriving on average,
demanding $E[S]$ server-seconds per job; the jobs are being served by $E[N_\busy^*]$ servers.
As a result, $\lambda^* = E[N_\busy^*]/E[S]$.
In contrast, if we ignored wasted servers,
we would overestimate $\lambda^*$ as $n/E[S]$,
an estimate that we call $\lambda^\naive$.
We can thus write wastage in terms of $\lambda^*$ and $\lambda^\naive$:
\[ \text{Wastage} = E[N_\idle^*] = n - E[N_\busy^*] = (\lambda^\naive - \lambda^*) E[S].\]
As a result,
we can also think of wastage as proportional to the gap between $\lambda^*$ and $\lambda^\naive$.

Wastage can have a major impact on response time (time from arrival to departure)
in multiserver-job systems.
\begin{figure}
    \includegraphics[width=0.5\textwidth]{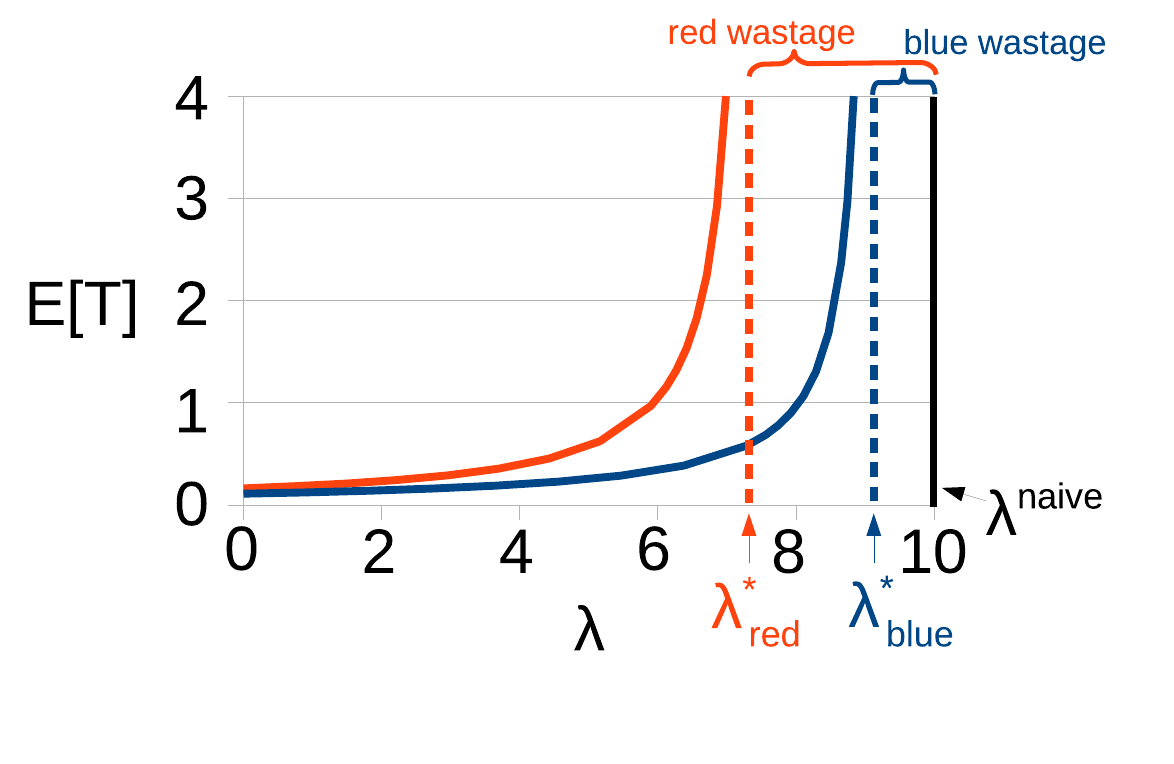}
    \vspace{-20pt}
    \caption{Mean response time $E[T]$ as a function of arrival rate $\lambda$,
    in two different systems: red and blue.
    In both systems, we set
    $n_1 = 1, n_2 = 10, \mu_1 = 2 \mu_2, n = 10,$ and $E[S] = 1$.
    In the blue system, $p_1 = 0.2, \mu_1 = 16.2, \mu_2 = 8.1$.
    In the red system, $p_1 = 0.6, \mu_1 = 8.6, \mu_2 = 4.3$.
    The black line shows $\lambda^\naive = 10$,
    where the edge of the stability region would lie in the absence of wastage.
    The dotted lines show $\lambda^*_\mathrm{blue}$ and $\lambda^*_\mathrm{red}$,
    the actual boundaries of the stability regions.
    Note that because $E[S] = 1$,
    the gap $\lambda^\naive - \lambda^*$ is equal to the wastage.}
    \label{fig:response_compare}
\end{figure}
We illustrate this impact in \cref{fig:response_compare},
where the solid red and blue curves
show the mean response time $E[T]$ as a function of
$\lambda$ in two different systems.
Both systems have the same mean server-seconds per job $E[S]=1$
and same number of servers $n=10$,
so both systems have the same naive stability region $\lambda^\naive=10$,
shown by the black line.
However, the two systems have very different amounts of wastage,
and hence very different values of $\lambda^*$,
shown by the dotted lines.
Moreover, these very different values of $\lambda^*$
shape two very different response time curves,
shown by the solid lines.

If we only had the simple estimate $\lambda^\naive$,
multiserver-job systems would be unpredictable and mysterious.
By deriving $\lambda^*$ for the two-class multiserver-job system,
we not only characterize wastage and stability,
but also take an important step towards understanding response time.

\subsection{Novel Perspective: Saturated vs. Non-saturated}
\label{sec:intro_perspective}
We solve the stability problem by shifting our focus:
Instead of directly analyzing the model described in \cref{sec:intro_model},
we start by analyzing an alternate model, the \emph{saturated system}.
In the saturated system there are always additional jobs in the queue,
so we never have to worry about
states where the queue is empty.
Instead, we can focus on only the states in which the servers are
as close to full as possible,
given the FCFS service policy.
This focus
enables us to derive a product-form steady state distribution
for the saturated system,
given in \cref{thm:embedded_saturated,thm:continuous_saturated}.

Next, we derive  \cref{thm:stability},
which characterizes
the stability region of the original model
in terms of the saturated system:
We show that $\lambda^*$,
the arrival rate which forms the upper boundary of the stability region of the original system,
is equal to the throughput of the saturated system.
Combining \cref{thm:embedded_saturated,thm:continuous_saturated,thm:stability}
allows us to characterize the stability region of our original model.

\subsection{Insights from our results}

Our analysis brings to light three important features of wastage
in multiserver-job systems, which are detailed in \cref{sec:lessons}.
\begin{enumerate}[(1)]
    \item A significant portion of the naive stability region can be lost to wastage,
        potentially 50\% or more.
        Wastage is at its worst when $n_2$ is close to or equal to $n$.
    \item Wastage is lower when jobs demanding fewer servers take less time,
        i.e. when $\mu_1 > \mu_2$.
        In practice, jobs demanding fewer servers typically do take less time,
        as seen in trace shown in \cref{fig:intro_correlation}.
        When $\mu_1$ and $\mu_2$ are roughly equal,
        or when $\mu_1 < \mu_2$, wastage is relatively higher.
    \item Wastage varies in complex and non-monotonic patterns.
        While (1) and (2) describe broad trends,
        these trends can temporarily run in reverse.
\end{enumerate}

\subsection{Contributions}

We derive the exact, closed-form stability region
for the two-class multiserver-job system.
This is the first analytical result for
any non-dropping multiserver-job system
where different classes of jobs have different service rates.
Moreover, our stability region result
has a simple sum-of-products formula,
because it is based on a product-form steady state result for the saturated system.

In addition, we use our saturated-system framework to give a dramatically
simpler proof of the stability region result of \cite{rumyantsev_stability},
in \cref{sec:simplify}.
Our proof illustrates
why the stability region of the model in \cite{rumyantsev_stability}
has a simple sum-of-products formula,
which is not addressed in \cite{rumyantsev_stability}.

Finally, in \cref{sec:lessons}, we use our solution
for the stability region to study wastage in our system.

\subsection{Outline}

In \cref{sec:prior_work}, we discuss the prior work on multiserver-job systems.
In \cref{sec:model}, we formally define our system model.
In \cref{sec:overview}, we overview our results,
deferring the proofs to
\cref{sec:embedded,sec:saturated,sec:standard}.
In \cref{sec:simplify},
we provide
a dramatically simpler analysis
of the single-service-time-distribution
model in \cite{rumyantsev_stability}.
In \cref{sec:lessons},
we discuss practical lessons.
Finally, in \cref{sec:discussion},
we discuss future directions.

\section{Prior work}
\label{sec:prior_work}
While FCFS multiserver-job models are not very common in the queueing literature,
there is some prior analytical and empirical work on related models,
which we describe in \cref{sec:dropping,sec:supercomputing,sec:vm_scheduling}.
In \cref{sec:queueing} we describe prior work in our FCFS multiserver-job model.

\subsection{Dropping multiserver-job model}
\label{sec:dropping}
Almost all existing analytical work in the multiserver-job model has focused
on a model where jobs which cannot immediately receive service
are dropped.
The \emph{dropping} multiserver-job model is well understood,
with steady state results known for very general settings.

\citet{arthurs_sizing} study a multiserver-job model with an arbitrary number of
job classes, each of which requires a fixed number of servers and
a different exponential distribution of service time.
\citet{arthurs_sizing} demonstrate that the steady state distribution follows a simple product form,
using a local balance argument.

Building upon the work of \citet{arthurs_sizing},
\citet{whitt_blocking} generalizes the dropping multiserver-job model
to allow jobs to demand two different types of server resources simultaneously.
In this generalized model,
Whitt derives another simple product-form solution
for the steady state distribution,
also making use of a local balance technique.

More recently, \citet{vandijk_blocking} generalizes the work of \citet{arthurs_sizing}
in a different direction,
allowing each job class to require a general service time distribution,
in contrast to the exponential service time distributions of prior work \citep{arthurs_sizing,whitt_blocking}.
In contrast to the prior models,
van Dijk considers a closed system,
where job completions
trigger new arrivals after a general ``think time'' distribution.
In this highly general setting,
\citet{vandijk_blocking} demonstrates
an insensitivity result,
showing that a product-form steady state distribution continues to hold.

Finally, \citet{tikhonenko_generalized}
combines the generality of \citep{whitt_blocking} and \citep{vandijk_blocking},
allowing two resources to be required as well as allowing a general service time distribution,
while returning to a Poisson arrival process.
Here the solution is far less simple,
but the author still derives a steady state solution.

All of the above works assume a dropping
multiserver-job model.

\subsection{Supercomputing}
\label{sec:supercomputing}
In supercomputing centers,
actual systems closely resemble the queueing multiserver-job model,
where jobs might demand anywhere from one core to thousands of concurrent cores 
\cite{msi_queues,vizino_batch}.
Unfortunately,
all the studies in this literature are simulation-based or empirical,
rather than analytical.
A particular focus area for these papers is studying
system utilization under a variety of scheduling policies,
such as FCFS, backfilling, and more novel policies.
Low utilization is the counterpart to a high number of wasted servers.

Many supercomputing papers have empirically observed
that utilization saturates well below 100\% efficiency
under FCFS scheduling \cite{feitelson_toward,jones_scheduling};
in some settings, FCFS utilization can be as low as 40\% \cite{jones_scheduling}.
Our findings in \cref{sec:number_wasted}
help explain this observation.

Due to the severity of wastage,
the supercomputing field is highly motivated to find ways to mitigate this behavior.
For example, reducing the maximum number of cores that any job can demand
has been observed to improve utilization \cite{jones_scheduling}.
Our findings in \cref{sec:wastage_falls}
help give an analytical explanation for this observation.

A wide variety of scheduling policies have been proposed
to alleviate the shortcomings of FCFS scheduling
\citep{tang_adaptive,kurowski_hierarchical,
feitelson_parallel,armstrong_scheduling,tang_reducing}.
These policies must juggle the tradeoff between fairness,
where FCFS excels,
and utilization, where FCFS is often lackluster.
Due to the lack of theoretical understanding of these systems,
extensive simulation is often performed to evaluate these policies
\citep{kurowski_hierarchical,tang_reducing}.

The extensive study of the multiserver-job model
in the supercomputing literature
motivates the need for theoretical studies of the model, such as ours.

\subsection{Virtual Machine scheduling}
\label{sec:vm_scheduling}
In the field of cloud computing, the Virtual Machine (VM) scheduling problem
is essentially a multi-resource generalization of the queueing multiserver-job model.
In the VM scheduling literature,
theoretical papers typically focus on finding
a throughput-optimal policy
\cite{maguluri_stochastic,psychas_nonpreemptive,guo_optimal,maguluri_scheduling,ghaderi_randomized},
and have achieved strong theoretical results in this direction.
However, their techniques are specific to throughput-optimal policies.
It is straightforward to characterize the stability region that a
throughput-optimal policy achieves,
and so these results focus on proving that a specific policy has that stability region,
which proves that the policy is throughput optimal.

In contrast, 
just characterizing the stability region under FCFS scheduling is highly nontrivial,
and the techniques developed for the throughput-optimal setting do not apply.
This is unfortunate,
because
the default scheduling policy used in the cloud-computing industry is the FCFS policy;
for example, FCFS is the default scheduling policy
in the CloudSim, iFogSim, CEPSim and GridSim cloud computing simulators \cite{madni_performance}.
However, despite the practical importance of FCFS scheduling,
comparisons with more advanced policies
have been limited to simulation \cite{cao_comparison,madni_performance}.

Our results take an initial step towards analytically characterizing
the performance of FCFS scheduling in the VM scheduling setting.
In doing so, we complement the work on throughput-optimal policies,
enabling an analytical comparison between 
these policies and FCFS.

\subsection{FCFS multiserver-job model}
\label{sec:queueing}
There are very few
analytical results for the FCFS multiserver-job model (without dropping).
All such results assume that jobs of any class
come from a \textit{single} job size distribution.

In 1979, \citet{kim_mms} studies an (FCFS) multiserver-job model
where jobs can demand any number of servers,
but all jobs require the same exponential distribution of service time.
Kim gives a matrix-geometric algorithm
to compute the steady state distribution of the number of jobs in the system.
Unfortunately, Kim's algorithm scales exponentially as the size of the system increases,
making it impractical for all but the smallest systems.

In 1984, \citet{brill_queues} focus on a two-server multiserver-job model,
with two classes that have the same exponential distribution of service time,
but require one or two servers respectively.
Brill and Green derive the steady state distribution of the system.
Unfortunately, their solution is highly complex,
involving the roots of a quartic equation.
As a result, their direct method does not easily generalize beyond the two-server system,
and does not provide much intuition.

In 2007, \citet{fillippopoulos_mm2} again study the two-server multiserver-job model,
with the same restriction that the two job classes
require the same exponential distribution of service time.
Fillippopoulos and Karatza also derive the steady state distribution of the system.
Like \citet{brill_queues}, their solution is highly complex,
requiring the roots of a similar quartic equation.
Due to the complexity of the solution,
Fillippopoulos and Karatza give simpler approximations
for mean queue length and mean response time.

Given the complexity of the steady state distribution for even a two-server model,
an attractive alternative approach is to characterize the stability region
of more general models.
In 2016, \citet{rumyantsev_stability} study a multiserver-job model
where jobs can demand any number of servers,
but again all jobs require the same exponential distribution of service time.
Rumyantsev and Morozov exactly characterize the stability region of the system
with a simple formula.
But despite the simple formula,
their proof technique provides little in the way of intuition for the solution.
In 2016, \citet{morozov_stability} generalize their model
to allow Markov Arrival Process arrivals, and show that the same stability region holds.

In 2019, \citet{afanaseva_stability} again study the stability region
of a multiserver job model.
They further generalize
the arrival model of \cite{morozov_stability}
to allow any regenerative arrival process.
They also generalize the service time distribution
to allow hypoexponential service times,
though they only derive an explicit solution for a simple two-server case.
Once again, all jobs must require the same distribution of service time.
To handle their highly general arrival process,
they introduce an ``auxiliary queueing system,''
which resembles our saturated system.

All prior analytical results in FCFS multiserver-job models
only consider systems in which all jobs
require the same distribution of service time.
This paper gives the first analytical characterization of stability
in a model without that assumption.

\section{System model}
\label{sec:model}

In this section we will define our system model and our notation.
First, we will introduce the standard, non-saturated system
in \cref{sec:non_saturated}.
This is referred to as the two-class multiserver-job model throughout this paper.
Afterwards, we will introduce the saturated system
in \cref{sec:saturated_model}.
We will first analyze the saturated system,
then use that analysis to solve the non-saturated system.

In both systems, we have two kinds of jobs: class 1 and class 2.
Class 1 jobs require $n_1$ servers and require $\Exp(\mu_1)$ time at each server;
$n_2$ and $\mu_2$ are defined similarly.
A job must be served concurrently at each of its servers
until the job is finished.
The system has $n$ servers in total.
We assume that $n_1 < n_2 \le n$,
so class 1 jobs require fewer servers.
We assume that an arriving job (or a job in the saturated queue)
is class 1 with i.i.d. probability $p_1$,
and is class 2 with probability $p_2 = 1-p_1$.
Finally, we assume that jobs are served in strict FCFS order.

\subsection{Non-saturated system}
\label{sec:non_saturated}
In the non-saturated system
jobs arrive over time, either entering service immediately or being queued.
We assume that jobs arrive according to a Poisson process with rate $\lambda$.
Our goal is to characterize the values of $\lambda$ for which the system is stable.

\begin{figure}
    \centering
    \includegraphics[width=0.5\textwidth]{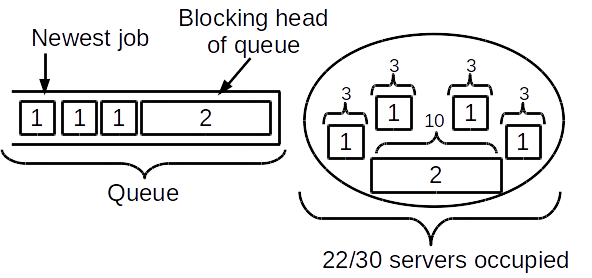}
    \caption{A possible state of the 3-10-30 system
    ($n_1 = 3$, $n_2 = 10$, and $n = 30$).
    Four class 1 jobs and one class 2 job
    are in service, a class 2 job is blocking the head of the queue,
    and three more class 1 jobs are in the queue.}
    \label{fig:queue_example}
\end{figure}
A crucial aspect of our analysis is choosing the right state descriptor for the system.
For example, consider a system with $n_1 = 3$, $n_2 = 10$, and $n = 30$.
We will call this the ``3-10-30 system.''
Suppose the 3-10-30 system is in the state shown in \cref{fig:queue_example}.
One straightforward way to describe the state
would be to 
list the classes of all jobs in the queue, in arrival order,
and count the number of jobs in each class at the server.
In \cref{fig:queue_example},
that descriptor would be $[[1, 1, 1, 2]; 4; 1]$, 
or in general
\begin{center}
[[ Queue ]; \# class 1 in service; \# class 2 in service ]
\end{center} 

However, this state descriptor contains some unnecessary information.
At this moment in time,
we do not need to know the classes of the three jobs at the back of the queue.
To simplify the state description,
we can delay sampling the classes of those jobs
until their classes become relevant.

The classes of the other jobs in the system
are all relevant in this example.
The classes of the jobs in service always matter,
because they determine the rate of job completions.
Here,
the class of the job at the front of the queue is also important:
Because the front job in the queue is a class 2 job,
demanding 10 servers, it cannot receive service,
while a class 1 job, demanding 3 servers,
would fit.
In situations
where a class 2 job does not fit but a class 1 job would fit,
we describe the class 2 job
as ``blocking the head of the queue,'' or ``blocking'' for short.

In contrast, if there were two more class 1 jobs in service,
there would only be 2 remaining servers available.
In that case, neither class of job would be able to fit.
In cases where neither class of job fits,
we say that all jobs in the queue are non-blocking.

A simpler state description
gives only the necessary information:
the number of non-blocking jobs in the queue
(of unspecified class),
the number (0 or 1) of blocking jobs,
and the number of jobs of each class in service.
For the example in \cref{fig:queue_example}, that descriptor is $[3, 1, 4, 1]$,
or in general
\begin{center}
    [\# (non-blocking) jobs in queue, \# blocking jobs, \# class 1 in service, \# class 2 in service]
\end{center}
We will use this state descriptor in this paper.
Note that by definition, a blocking job must be a class 2 job,
and there can be at most one such job in the system.

\subsection{Saturated System}
\label{sec:saturated_model}

To describe a state in the saturated system,
we simply omit the number of non-blocking jobs in the queue,
and otherwise use the same descriptor as the non-saturated system.
For instance,
the saturated system state corresponding to
\cref{fig:queue_example}
would be $[1, 4, 1]$.

Notice that in the saturated system,
there are only a finite number of possible states.
For instance, in the 3-10-30 system from \cref{fig:queue_example},
the possible states are:
\begin{align*}
    [0, 0, 3], [1, 1, 2], [1, 2, 2], [0, 3, 2],
    [1, 4, 1], [1, 5, 1],
    [0, 6, 1], [1, 7, 0],
    [1, 8, 0], [1, 9, 0], [0, 10, 0].
\end{align*}

Note that there is \emph{exactly one state for each possible number of class 1 jobs in service}.
To see why, imagine starting with a given number of class 1 jobs in service,
then adding class 2 jobs until either
the servers are ``full'' (there are not enough servers for an additional job of either class)
or the queue is blocked.
This will always result in a unique state of the saturated system.
We will refer to the unique state with exactly $a$ class 1 jobs in the server
as $s_1(a)$. This state can be blocking or non-blocking.
In the 3-10-30 system, $s_1(4)$ is the state $[1, 4, 1]$,
while $s_1(6)$ is the state $[0, 6, 1]$.

There can be many states with a given number of class 2 jobs in service,
but there is always \emph{a unique state of the system
with a given number of class 2 jobs in service and no job blocking the queue}.
To see why, imagine starting with a given number of class 2 jobs in service,
then adding class 1 jobs until
no more class 1 jobs will fit into service.
This will always result in a unique non-blocked state of the saturated system.
We will refer to the unique non-blocked state with exactly $b$ class 2 jobs in the server
as $s_2(b)$.
For instance, in the 3-10-30 system, $s_2(2)$ is the state $[0, 3, 2]$.

We will analyze two different Markov chains based on the saturated system.
First, there is the standard continuous-time Markov chain (CTMC).
Second, we will analyze the embedded DTMC
where we only record the state just after each job completion.
The CTMC is characterized by the rate at which transitions occur from each state,
as well as the probability of transitioning to each other state.
In the embedded DTMC,
we only need to think about the transition probabilities.

Every transition consists of a job completion,
followed by zero, one or many jobs moving from the infinite queue
into service,
and possibly a class 2 job blocking the head of the queue.

The rate of transitions out of each state
must equal the rate of completions in that state.
In state $[\cdot, a, b]$, this rate is $a \mu_1 + b \mu_2$.
For instance, in the 3-10-30 system, in state $[1, 4, 1]$,
completions occur at a rate of $4\mu_1 + \mu_2$.
Note that we consider every completion
to be a transition,
even if a job enters from the queue which matches the completed job,
resulting in the state not changing.
For instance, if we are in state $[0, 10, 0]$
in the 3-10-30 system,
and a class 1 job completes and another class 1 job enters from the queue,
we consider this event to be a transition.

Every transition can be specified by the class of job
which completes,
and the classes of jobs which enter from the queue.
For a given state $s$, let $f_1(s)$ denote the probability that the
next completion is a class 1 job,
and let $f_2(s) = 1-f_1(s)$ denote the probability the next completion
is a class 2 job.
For a given state $[\cdot, a, b]$, we have
\begin{align}
    \label{eq:def_f}
    f_1([\cdot, a, b]) = \frac{a \mu_1}{a \mu_1 + b \mu_2},
    \qquad f_2([\cdot, a, b]) = \frac{b \mu_2}{a \mu_1 + b \mu_2}.
\end{align}

To calculate the probability of a specific transition occurring,
there are four factors to consider:
\begin{itemize}
    \item The initial state $s$.
    \item The class of the completing job: $i$.
    \item The number of class 1 jobs which must enter from the queue: $j$.
    \item The number of class 2 jobs which must enter from the queue: $k$.
\end{itemize}
Observe that $k \le 1$.
More specifically, observe that if a class 2 job enters from the queue,
as opposed to from a blocking position,
that job must be the last job to enter from the queue in a given transition.
As a result, for a given transition there is only one possible order
in which jobs can enter from the queue.
Thus, every transition has a probability of occurring of the form
$f_i(s) p_1^j p_2^k.$

Generically, we refer to the probability of transitioning
from state $s$ to state $s'$ as $P(s, s')$.
For example, in the 3-10-30 system, from \cref{fig:queue_example},
\[P([1, 4, 1], [0, 6, 1]) = f_2([1, 4, 1]) p_1^2\]
since we can only transition from $[1, 4, 1]$ to $[0, 6, 1]$
if a class 2 job completes,
the blocking job moves into service,
and the next two jobs in the queue are class 1 jobs.

\section{Overview of results}
\label{sec:overview}
Our goal is to derive the stability region
of the (non-saturated) multiserver-job system.
In order to do so, we start by deriving the throughput of the saturated system.

We will first study the embedded DTMC of the saturated system,
where we only look at the state after each departure (\cref{thm:embedded_saturated}).
This result will then enable us to derive the throughput of the CTMC of the saturated system
(\cref{thm:continuous_saturated}).
We will then use the throughput of the saturated system
to find the stability region of the original non-saturated system
(\cref{thm:stability}).

Our first theorem gives the steady state
of the embedded chain of the saturated system:
\begin{theorem}
    \label{thm:embedded_saturated}
    The steady state distribution of the embedded DTMC of the saturated system is:
    \begin{align}
        \nonumber
        \pi_{[0, a, b]} &= C p_1^a p_2^b
            \prod_{i=1}^a \frac{1}{f_1(s_1(i))}
            \prod_{j=1}^b \frac{1}{f_2(s_2(j))}\\
            \label{eq:steady_state_guess}
        \pi_{[1, a, b]} &= C p_1^a p_2^{b+1}
            \prod_{i=1}^a \frac{1}{f_1(s_1(i))}
            \prod_{j=1}^b \frac{1}{f_2(s_2(j))}
    \end{align}
    where $C$ is a normalizing constant.
\end{theorem}

From \cref{thm:embedded_saturated},
we derive the steady state and throughput of the
CTMC of the saturated system:
\begin{theorem}
    \label{thm:continuous_saturated}
    The steady state distribution of the saturated system is:
    \[ \p_{[h, a, b]} = X \frac{\pi_{[h, a, b]}}{a \mu_1 + b \mu_2}\]
    where $X$, the normalizing constant,
    is the throughput of the saturated system:
    \[ X = \left( \sum_{[h, a, b]} \frac{\pi_{[h, a, b]}}{a \mu_1 + b \mu_2} \right)^{-1}.\]
\end{theorem}
Finally, we use \cref{thm:continuous_saturated} to derive the stability region
of the original non-saturated system.
\begin{theorem}
    \label{thm:stability}
    The original non-saturated system is stable with arrival rate $\lambda$ if
    \[ \lambda < \lambda^* \]
    where $\lambda^* = X$ is the throughput of the saturated system,
    and unstable if $\lambda > \lambda^*$.
\end{theorem}

We will prove \cref{thm:embedded_saturated} in \cref{sec:embedded},
\cref{thm:continuous_saturated} in \cref{sec:saturated},
and \cref{thm:stability} in \cref{sec:standard}.

\section{Embedded chain of saturated system: Proof of Theorem~\ref{thm:embedded_saturated}}
\label{sec:embedded}
\begin{proof}[Proof of \cref{thm:embedded_saturated}]

    In order to verify the guess in \cref{eq:steady_state_guess},
    we must show that the stationary equations hold for each state:
    \begin{align}
        \label{eq:balance}
        \pi_s = \sum_{s'} \pi_{s'} P(s', s)
    \end{align}
    For instance, take the 3-10-30 system,
    with $n_1 = 3$, $n_2 = 10$ and $n = 30$.
    Let's write down the stationary equation for state $s = [1, 5, 1]$.
    First, the possible states that can transition to state $s$ are
    \[ [0, 6, 1], [0, 3, 2], [1, 4, 1], [1, 5, 1] \]
    Therefore, \cref{eq:balance} states that
    we must show that under the steady state guess,
    \begin{align*} \pi_{[1, 5, 1]} &= \pi_{[0, 6, 1]} P([0, 6, 1], [1, 5, 1])
        + \pi_{[0, 3, 2]} P([0, 3, 2], [1, 5, 1])\\
        &+ \pi_{[1, 4, 1]} P([1, 4, 1], [1, 5, 1])
        + \pi_{[1, 5, 1]} P([1, 5, 1], [1, 5, 1])\\
        &= \pi_{[0, 6, 1]} f_1([0, 6, 1]) p_2
        + \pi_{[0, 3, 2]} f_2([0, 3, 2]) p_1^2 p_2\\
        &+ \pi_{[1, 4, 1]} f_2([1, 4, 1]) p_1 p_2
        + \pi_{[1, 5, 1]} f_2([1, 5, 1]) p_2
    \end{align*}

    In order to simplify the process of proving that the stationary equations hold,
    we will decompose each stationary equation
    into two simpler equations,
    corresponding to the class of the job completed in each transition.
    By proving that each decomposed equation holds,
    we will prove that the their sum, the stationary equation, also holds.
    
    Let $P_1(s', s)$ be the probability
    a transition from $s'$ to $s$
    due to a completion of a class 1 job,
    and define $P_2(s', s)$ similarly.
    Note that $P_1(s', s) + P_2(s', s) = P(s', s)$.

    Let us define a pair of decomposed stationary equations for any state $s$:
    \begin{align}
        \label{eq:decompose_1}
        p_1 \pi_s = \sum_{s'} \pi_{s'} P_1(s', s)
        \qquad
        \qquad
        p_2 \pi_s = \sum_{s'} \pi_{s'} P_2(s', s)
    \end{align}
    The decomposed stationary equations in \cref{eq:decompose_1}
    are sufficient, though not necessary,
    for the stationary equation \cref{eq:balance} to hold.
    For a specific example, in the state $s = [1, 5, 1]$ in the 3-10-30 system,
    the decomposed stationary equations require that
    \begin{align*}
        p_1 \pi_{[1, 5, 1]} &= \pi_{[0, 6, 1]} f_1([0, 6, 1]) p_2\\
        p_2 \pi_{[1, 5, 1]} &= \pi_{[0, 3, 2]} f_2([0, 3, 2]) p_1^2 p_2
        + \pi_{[1, 4, 1]} f_2([1, 4, 1]) p_1 p_2
        + \pi_{[1, 5, 1]} f_2([1, 5, 1]) p_2
    \end{align*}

    Towards proving that the decomposed stationary equations hold,
    we will first enumerate the transitions which are possible from each state,
    in \cref{sec:possible_transitions}.
    Then, we will use that enumeration to verify the decomposed stationary equations
    in the rest of \cref{sec:embedded}.

    Specifically, we will prove that the decomposed stationary equations hold
    in five cases, collectively covering all states and all decomposed stationary equations:
    \begin{itemize}
        \item The class 1 decomposed equation for states without a blocking job,
            i.e. $[0, a, b]$: see \cref{sec:class_1_no_blocking}.
        \item The class 1 decomposed equation for states with a blocking job,
            i.e. $[1, a, b]$: see \cref{sec:short_class_1_blocking}.
        \item The class 2 decomposed equation for $[0, a, b]$ states:
            see \cref{sec:class_2_no_blocking}.
        \item The class 2 decomposed equation for $[1, a, b]$ states:
            see \cref{sec:short_class_2_blocking}.
        \item Edge cases where $a$ or $b$ is 0: see \cref{sec:short_edge_cases}.
    \end{itemize}
    \subsection{Possible transitions}
    \label{sec:possible_transitions}
    \begin{figure}
        \centering
        \includegraphics[width=\textwidth]{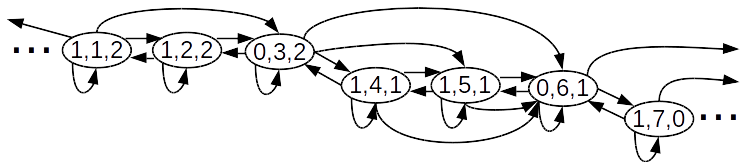}
        \caption{All possible transitions in the 3-10-30 system,
        when starting from states with 1 to 7 class 1 jobs.}
        \label{fig:possible_transitions}
    \end{figure}
    In \cref{lem:possible_transitions},
    we enumerate all transitions which are possible from each state in a generic system.
    To help visualize this,
    \cref{fig:possible_transitions}
    shows all possible transitions for part of the 3-10-30 system.

    \begin{lemma}
        \label{lem:possible_transitions}
        The possible transitions in the saturated system are exactly:
        \begin{enumerate}[i.]
            \item From $[0, a, b]$ to $[0, a, b]$, via a class 1 completion, whenever $a > 0$.
            \item From $[0, a, b]$ to $s_1(a-1)$, via a class 1 completion, whenever $a > 0$.
            \item From $[0, a, b]$ to $[0, a, b]$, via a class 2 completion, whenever $b > 0$.
            \item From $[0, a, b]$ to $s_2(b-1)$, via a class 2 completion, whenever $b > 0$.
            \item From $[0, a, b]$ to $[1, a', b-1]$,
                where $a'$ is any number greater than $a$
                and less than the number of class 1 jobs in state $s_2(b-1)$.
                Occurs via a class 2 completion, can occur whenever $b>0$.
            \item From $[1, a, b]$ to $s_1(a-1)$, via a class 1 completion,
                whenever $a > 0$.
            \item From $[1, a, b]$ to $s_2(b)$, via a class 2 completion,
                whenever $b > 0$.
            \item From $[1, a, b]$ to $[1, a', b]$,
                where $a'$ is any number greater than or equal to $a$
                and less than the number of class 1 jobs in state $s_2(b)$.
                Occurs via a class 2 completion, can occur whenever $b > 0$.
        \end{enumerate}
    \end{lemma}
    \begin{proof}
    We will first handle the case of starting from non-blocking states,
    then the case of starting from blocking states.
    \subsubsection{Non-blocking states}
    First, consider starting from the state $[0, a, b]$.
    Out of the $n$ servers, at most $n_1 - 1$ can be unoccupied.
    Otherwise,
    a job would either enter service or block the queue.

    Every transition begins with a job completion.
    The job which completes can be either a class 1 job or a class 2 job
    (unless $a$ or $b$ is 0).
    If a class 1 job completes, then between $n_1$ and $2n_1 - 1$ servers are unoccupied.
    At this point, another job must enter from the queue.
    If that job is a class 1 job, then we are back in state $[0, a, b]$,
    and our transition is complete.

    If that job is a class 2 job, then either the job enters service,
    leaving fewer than $n_1 -1$ servers unoccupied,
    or the job blocks the head of the queue.
    In either case, the transition is complete.
    The state transitioned to is $s_1(a-1)$,
    which may be either $[0, a-1, b+1]$ or $[1, a-1, b]$.
    In the 3-10-30 system, the transition from $[0, 3, 2]$ to $[1, 2, 2]$
    would fall into this category.

    On the other hand, if we are in state $[0, a, b]$ and a class 2 job completes,
    between $n_2$ and $n_2 + n_1 - 1$ servers are unoccupied.
    At this point, a job must enter from the queue.
    If that job is a class 2 job, then we are back in state $[0, a, b]$,
    and our transition is complete.

    If that job is a class 1 job, then at most $n_2 - 1$ servers are unoccupied.
    At this point, more jobs may enter from the queue.
    In the case where only class 1 jobs
    arrive from the queue,
    filling the servers,
    the system transitions to state $s_2(b-1)$,
    since no job ends up blocking the head of the queue.
    In the 3-10-30 system, the transition from $[0, 3, 2]$ to $[0, 6, 1]$
    falls into this category.

    On the other hand, a class 2 job may arrive after
    at least one class 1 job has arrived but
    before the servers are full.
    The class 2 job cannot fit into service, since there were previously at most $n_2 - 1$
    unoccupied servers.
    Instead, the class 2 job will block the head of the queue.
    Transitions in this category can take us to any state of the form $[1, a', b-1]$,
    where $a'$ is greater than $a$ but less than the number of
    class 1 jobs in the state $s_2(b-1)$.
    In the 3-10-30 system, the transitions from $[0, 3, 2]$ to $[1, 4, 1]$ and $[1, 5, 1]$
    fall into this category.

    That completes the proof for states of the form $[0, a, b]$.
    \subsubsection{Blocking states}

    Next, consider starting from the state $[1, a, b]$.
    Out of the $n$ servers available, between $n_1$ and $n_2 - 1$ servers are
    unoccupied.
    Any more, and the blocking job would enter service.
    Any fewer, and the job at the head of the queue would not be a blocking job.

    If a class 1 job completes,
    then between $2n_1$ and $n_1 + n_2 - 1$ servers are left unoccupied.
    If fewer than $n_2$ servers are left unoccupied,
    then there is not enough room for the job blocking the head of the queue
    to enter service, and the transition is complete.
    The state transitioned to is $[1, a-1, b]$.
    In the 3-10-30 system, the transition from $[1, 5, 1]$ to $[1, 4, 1]$
    falls into this category.

    Otherwise, the class 2 job enters service,
    leaving at most $n_1 - 1$ servers unoccupied.
    This is not enough room for another job,
    so the transition is complete.
    The state transitioned to is $[0, a-1, b+1]$.
    In the 3-10-30 system, the transition from $[1, 4, 1]$ to $[0, 3, 2]$
    falls into this category.

    These two cases,
    transitioning from $[1, a, b]$ to either $[1, a-1, b]$ or $[0, a-1, b+1]$,
    can be summarized as a transition to state $s_1(a-1)$.
    
    On the other hand, if we are in state $[1, a, b]$ and a class 2 job completes, 
    the job blocking the head of the queue will enter service.
    At this point, between $n_1$ and $n_2 - 1$ servers are left unoccupied,
    so more jobs must enter from the queue.
    In the case where only class 1 jobs
    arrive from the queue,
    filling the servers,
    the system transitions to state $s_2(b)$.
    In the 3-10-30 system, the transition from $[1, 4, 1]$ to $[0, 6, 1]$
    falls in this category.

    The other possibility is that a class 2 job may arrive from the queue,
    possibly after some class 1 jobs have arrived.
    The class 2 job must block the head of the queue,
    since there were previously at most $n_2 - 1$ unoccupied servers.
    Transitions in this category can take us to any state of the form $[1, a', b]$,
    where $a'$ is at least $a$ and less than the number of class 1 jobs in state $s_2(b)$.
    In the 3-10-30 system, the transitions from $[1, 4, 1]$ to $[1, 4, 1]$ and $[1, 5, 1]$
    fall into this category.

    That completes the proof for states of the form $[1, a, b]$,
    and hence the proof of \cref{lem:possible_transitions}.
    \end{proof}
    \subsection{Class 1 completions, no blocking job}
    \label{sec:class_1_no_blocking}
    Starting in state $[0, a, b]$,
    we want to show that the class 1 decomposed stationary equation holds:
    \[ p_1 \pi_{[0, a, b]} = \sum_{s'} \pi_{s'} P_1(s', [0, a, b]). \]
    To start with, we enumerate the states that can transition to $[0, a, b]$
    via the completion of a class 1 job.
    By inspecting \cref{lem:possible_transitions},
    we can see that transitions to non-blocking states
    via class 1 completions
    only occur in cases $i, ii,$ and $vi$.
    Case $i$ corresponds to the predecessor state $[0, a, b]$,
    while cases $ii$ and $vi$
    correspond to the potential predecessor states $[0, a+1, b-1]$
    and $[1, a+1, b-1]$, respectively.
    Only one of these states exists in a given system,
    and that state is referred to as $s_1(a+1)$.
    By \cref{lem:possible_transitions}, states $[0, a, b]$ and $s_1(a+1)$
    are the only possible states
    that could transition to $[0, a, b]$
    after a class 1 job completes.

    We will prove that the decomposed stationary equation
    holds both when $s_1(a+1)$ is $[0, a+1, b-1]$,
    and when $s_1(a+1)$ is $[1, a+1, b-1]$.

    To prove the steady state guess, we must show that
    \begin{align}
        \label{eq:class_1_setup}
        p_1 \pi_{[0, a, b]} = \pi_{[0, a, b]} P_1([0, a, b], [0, a, b]) +
        \pi_{s_1(a+1)} P(s_1(a+1), [0, a, b]).
    \end{align}
    Note that $P_1([0, a, b], [0, a, b]) = f_1([0, a, b]) p_1$, so \cref{eq:class_1_setup}
    simplifies to
    \begin{align}
        \label{eq:class_1_finish}
        p_1 \pi_{[0, a, b]} f_2([0, a, b]) = \pi_{s_1(a+1)} P(s_1(a+1), [0, a, b]).
    \end{align}
    Note that $P([0, a+1, b-1], [0, a, b]) = f_1([0, a+1, b-1]) p_2$
    and that $P([1, a+1, b-1], [0, a, b]) = f_1([1, a+1, b-1])$.
    By examining the steady state guess \cref{eq:steady_state_guess},
    we can see that in either case,
    \cref{eq:class_1_finish} must hold.
    As a result, the decomposed stationary equation must hold as well.

    \subsection{Class 1 completions, blocking job}
    \label{sec:short_class_1_blocking}
    This case is similar to the case in \cref{sec:class_1_no_blocking}
    and is proven in \cref{sec:class_1_blocking}.
    \subsection{Class 2 completions, no blocking job}
    \label{sec:class_2_no_blocking}
    Starting in state $[0, a, b]$,
    we want to show that the class 2 decomposed stationary equation holds:
    \[ p_2 \pi_{[0, a, b]} = \sum_{s'} \pi_{s'} P_2(s', [0, a, b]). \]

    To start with, we enumerate the states that can transition to $[0, a, b]$
    via the completion of a class 2 job.
    By inspecting \cref{lem:possible_transitions},
    we can see that transitions to non-blocking states via class 2 completions
    only occur in cases $iii, iv,$ and $vii$.
    Case $iii$ corresponds to the predecessor state $[0, a, b]$,
    with transition probability $f_2([0, a, b]) p_2$.
    Case $iv$ corresponds to the predecessor state $s_2(b+1)$.
    If we write $s_2(b+1)$ as $[0, a-i^*, b+1]$ for some $i^*$,
    then this case has transition probability $f_2(s_2(b+1)) p_1^{i^*}$.
    Case $vii$ corresponds to the set of predecessor states of the form
    $[1, a-i, b]$, where $1 \le i < i^*$,
    with transition probabilities $f_2([1, a-i, n]) p_1^i$.

    To prove the decomposed stationary equation,
    we must show that
    \begin{align*}
        p_2 \pi_{[0, a, b]} = \pi_{[0, a-i^*, b+1]} f_2([0, a-i^*, b+1]) p_1^{i^*}
            + \pi_{[0, a, b]} f_2([0, a, b]) p_2
            + \sum_{i=1}^{i^*-1} \pi_{[1, a-i, b]} f_2([1, a-i, b]) p_1^i\\
        \iff p_2 \pi_{[0, a, b]} f_1([0, a, b])
        = \pi_{[0, a-i^*, b+1]} f_2([0, a-i^*, b+1]) p_1^{i^*}
        + \sum_{i=1}^{i^*-1} \pi_{[1, a-i, b]} f_2([1, a-i, b]) p_1^i
    \end{align*}
    Let us apply an auxiliary lemma:
    \begin{lemma}
        \label{lem:telescope}
        For all $a, b$ and all $q \le r$ such that both $[1, a-q, b]$
        and $[1, a-r, b]$ are valid states,
        under the steady state guess in \cref{eq:steady_state_guess},
        \[
            \sum_{i=q}^r p_1^i \pi_{[1, a-i, b]} f_2([1, a-i, b])
            = p_1^q \pi_{[1, a-q, b]} - p_1^r \pi_{[1, a-r, b]}f_1([1, a-r, b]).
        \]
    \end{lemma}
    \begin{proof}
    Deferred to \cref{sec:proof_telescope}.
    \end{proof}
    We apply \cref{lem:telescope}
    with $q=1$ and $r = i^*-1$.
    After doing so, our desired statement simplifies to
    \begin{align}
        \label{eq:step_after_telescope}
        p_2 \pi_{[0, a, b]} f_1([0, a, b]) &=
        \pi_{[0, a-i^*, b+1]} f_2([0, a-i^*, b+1]) p_1^{i^*}\\
        &+ p_1\pi_{[1, a-1, b]}
        -p^{i^*-1}\pi_{[1,a-(i^*-1),b]}f_1([1,a-(i^*-1),b]).
        \nonumber
    \end{align}
    By comparing the steady state guesses for states $[0, a-i^*, b+1]$
    and $[1, a-i^*+1, b]$, we see that
    \begin{align}
        \label{eq:last_step}
        \pi_{[0, a-i^*, b+1]} f_2([0, a-i^*, b+1]) p_1 = \pi_{[1, a-i^*+1, b]} f_1([1, a-i^*+1, b]).
    \end{align}
    Substituting \cref{eq:last_step} into \cref{eq:step_after_telescope},
    the two corresponding terms cancel,
    so we only need to prove that
    $p_2 \pi_{[0, a, b]} f_1([0, a, b]) = p_1\pi_{[1, a-1, b]}$.
    This follows immediately from the steady state guess,
    by comparing the expressions for states $[0, a, b]$ and $[1, a-1, b]$.
    Thus, the decomposed stationary equation holds in this case.

    \subsection{Class 2 completions, blocking job}
    \label{sec:short_class_2_blocking}
    This case is similar to the case in \cref{sec:class_2_no_blocking}
    and is proven in \cref{sec:class_1_blocking}.
    \subsection{Edge cases}
    \label{sec:short_edge_cases}
    In the preceding cases, we verified that the decomposed stationary equations hold,
    but in doing so we assumed that states of the form $s_1(a+1)$ and $s_2(b+1)$ exist.
    In certain edge-case states, this assumption does not hold,
    so we must verify that the decomposed stationary equations still hold.

    The main additional fact we will use is that
    if jobs of only one class are in service, then a job of that class must complete next.
    In other words, $f_1([\cdot, a, 0]) = 1$ and $f_2([\cdot, 0, b]) = 1$.
    The proof of these cases is deferred to \cref{sec:edge_cases}.

    \subsection{Completing the proof of \cref{thm:embedded_saturated}}

    Combining
    \cref{sec:class_1_no_blocking,sec:class_2_no_blocking,sec:short_class_1_blocking,sec:short_class_2_blocking,sec:short_edge_cases},
    we have now proven that the decomposed stationary equations must hold in every state.
    Thus, we have proven that the stationary equations hold in every state.
    Therefore, \cref{eq:steady_state_guess} gives the stationary distribution for the
    embedded chain of the saturated system.
\end{proof}

\section{Continuous-time saturated system: Proof of Theorem~\ref{thm:continuous_saturated}}
\label{sec:saturated}
\cref{thm:continuous_saturated} follows from a generic transformation between
the steady state of the embedded chain of a CTMC
and the steady state of the CTMC itself.
The stationary probability $\p_s$ of each state $s$ in the CTMC
is simply
the probability $\pi_s$ in the embedded DTMC
divided by the rate of departures from state $s$.
The formal proof is deferred to \cref{sec:saturated_proof}.

\section{Stability region: Proof of Theorem~\ref{thm:stability}}
\label{sec:standard}
\begin{reptheorem}{thm:stability}
    The original non-saturated system is stable with arrival rate $\lambda$ if
    $\lambda < \lambda^*$,
    where $\lambda^* = X$ is the throughput of the saturated system,
    and unstable if $\lambda > \lambda^*$.
\end{reptheorem}
    \subsection{Proof sketch}
    The main part of the proof will show that when $\lambda < X$,
    the non-saturated system is positive recurrent.
    Subsequently, we will show that when $\lambda > X$,
    the non-saturated system is transient.

    Assuming that $\lambda < X$,
    define the set $\mathcal{E}$
    to consist of the states in the non-saturated system
    with no non-blocking jobs in the queue.
    We call the states in $\mathcal{E}$ the ``near-empty'' states.
    Our goal is to prove that the system returns to $\mathcal{E}$
    with probability 1, and in finite mean time.
    We say that $\mathcal{E}$ is a positive-recurrent set if this property holds.

    In order to show that $\mathcal{E}$ is a positive-recurrent set,
    we define two additional coupled systems.
    We first define the Augmented Saturated System (AugSS),
    which consists of the ordinary saturated system
    and a counter whose value mirrors the number of jobs
    in the non-saturated system.

    From here, we will derive a long period of time $t_1$,
    such that the expected number of completions 
    in the AugSS over any interval of length $t_1$
    exceeds $\lambda t_1$.
    We will use this interval to define an embedded DTMC
    for the AugSS which samples the state every $t_1$ time steps.

    We will show that this embedded DTMC is positive recurrent,
    which in turn implies that the Augmented Saturated System is positive recurrent,
    and so $\mathcal{E}$ must be a positive-recurrent set for the original non-saturated system.
    
    \subsection{Proof of \cref{thm:stability}}

    \begin{proof}
        Let us assume that $\lambda < X$.
        We want to show that the non-saturated system is positive recurrent,
        by showing that $\mathcal{E}$ is a positive-recurrent set.

        We start by defining the Augmented Saturated System (AugSS),
        which consists of two parts:
        A saturated system, as described in \cref{sec:saturated_model},
        and an additional ``jobs counter.''
        The jobs counter increments according to a Poisson process with rate $\lambda$,
        and decrements when jobs complete in the saturated system.
        As an edge condition, the jobs counter cannot decrement below 0.

        We couple the AugSS to the original non-saturated system
        by matching up their departure and arrival processes.
        Whenever the two systems
        have the same numbers of jobs of each class in service,
        we couple the same class of job to depart simultaneously in both systems.
        Otherwise, departures occur independently in the two systems.
        If completions occur in both systems,
        requiring the sampling of the classes of jobs
        in both systems for entrance into service,
        we couple the sampling so that the jobs' classes match in the two systems.
        Furthermore, we couple together the arrivals
        to the non-saturated system and the increments of the job counter.

        With this coupling in place, we inaugurate the AugSS
        at the moment when the non-saturated system exits the set $\mathcal{E}$.
        At this point, we set the state of the saturated subsystem
        to match the jobs in the non-saturated system at the servers and 
        blocking the head of the queue.
        We also inaugurate the jobs counter to match the total number of jobs
        in the non-saturated system.

        From this point in time onwards, as long as the non-saturated system
        remains outside of the set $\mathcal{E}$,
        the set of jobs at the servers and blocking the head of the queue
        will remain identical in the non-saturated system
        and in the AugSS.
        Likewise, the value of the jobs counter
        will match the total number of jobs in the non-saturated system.

        To prove that $\mathcal{E}$ is a positive-recurrent set
        for the non-saturated system,
        it suffices to show that the set of states with job-counter equal to zero
        is a positive-recurrent set for the AugSS.
        By the time the job counter reaches zero,
        the non-saturated system must have entered a state in $\mathcal{E}$.

        Next, 
        we derive the period of time $t_1$
        that we will use to define the embedded DTMC of the AugSS.
        Because the saturated system
        is a finite-state CTMC,
        it must converge to its steady state.
        In particular, for any $\epsilon > 0$ there exists a time $t_0 := t_0(\epsilon)$ such that
        regardless of the starting state,
        at any time $t > t_0$,
        the probability that the saturated system
        is in a given state $s$ is within $\epsilon$ of the steady-state probability $\pi_s$.
        As a result, the expected completion rate over any interval which begins after time $t_0$
        must exceed
        $X - f(\epsilon),$
        where $f$ is a function that vanishes as $\epsilon$ approaches 0.

        In particular, suppose we choose $\epsilon$ such that 
        $X - f(\epsilon) = \lambda  + \delta$
        for some $\delta > 0$.
        Such an $\epsilon$ must exist, because $X > \lambda$
        and $f(\epsilon) \rightarrow 0$ as $\epsilon \rightarrow 0$.
        Then at a given time $t > t_0$,
        the expected number of jobs completed in the saturated system is at least
        $(t - t_0) (\lambda + \delta).$

        Let $C(t)$ denote the number of jobs completed in the saturated system
        by time $t$,
        and let $A(t)$ denote the number of jobs which arrive to the non-saturated system
        by time $t$ (or equivalently, the number of times the jobs counter increments).
        Let $t_1$ be the time
        \[ t_1 = \frac{t_0 (\lambda+\delta) + 1}{\delta}. \]
        Then we can lower bound the expected number of jobs completed by time $t_1$:
        \begin{align*}
            E[C(t_1)] &\ge (t_1 - t_0) (\lambda + \delta)
            = \left( \frac{t_0 (\lambda+\delta) + 1}{\delta} - t_0 \right)(\lambda + \delta)
            = \frac{\lambda t_0 (\lambda + \delta) + \lambda + \delta}{\delta}
            = \lambda t_1 + 1
        \end{align*}
        Thus, $E[C(t_1)]$ is strictly greater than $E[A(t_1)]$,
        which equals $\lambda t_1$.
        More generally, because the above argument holds when starting from an
        arbitrary initial state,
        the expected number of completions over any interval of length $t_1$
        must exceed the expected number of arrivals over that interval.

        Let us define the embedded Markov chain of the AugSS,
        or the ``embedded chain,''
        to sample the
        AugSS's state every $t_1$ time steps.
        The expected number of completions between every update of this DTMC
        exceeds the expected number of arrivals.
        We will use this property of the embedded chain
        to employ Foster's theorem \cite{foster_stochastic},
        which will show that
        the embedded chain is positive recurrent.
        In particular, we will use the value of the jobs counter
        as the Lyapunov function for Foster's theorem.
        
        \begin{lemma}
            \label{lem:foster}
            Let $V$ be Lyapunov function which maps each state in the
            embedded chain
            to the value of its jobs counter.
            Then $V$ satisfies the conditions of Foster's theorem,
            showing that the embedded chain is positive recurrent.
        \end{lemma}
        \begin{proof}
            Deferred to \cref{app:foster}.
        \end{proof}

        Because the embedded chain is positive recurrent,
        the AugSS must also be positive recurrent.
        This implies that $\mathcal{E}$ forms a positive-recurrent set
        for the original non-saturated system, as desired.
        
        To show that the non-saturated system is unstable for $\lambda > X$,
        we can use a very similar argument.
        The only difference is that
        we can now show that every $t_1$ time steps,
        the expected number of arrivals exceeds the expected number of arrivals.
        As a result, the embedded DTMC for the AugSS is unstable,
        and so the AugSS is unstable as well,
        and so the original system is also unstable.
\end{proof}

\section{A simpler proof of Rumyantsev and Morozov \cite{rumyantsev_stability}}
\label{sec:simplify}

Our method of analyzing the stability region of multiserver-job systems
can be applied beyond the two-class model studied in this paper.
In this section,
we show how our method can be used to derive the stability region
of the multiserver-job model in \cite{rumyantsev_stability},
where jobs can require any number of servers,
but all jobs require the same exponential distribution of service time,
regardless of the number of servers required.
Our method has the advantage of providing clearer intuition
into the nature of the stability region.
Moreover, we will essentially reuse \cref{thm:continuous_saturated,thm:stability},
allowing the portion of the proof which is specific to the model
from \cite{rumyantsev_stability}
to be considerably simpler.

The system in \cite{rumyantsev_stability} has $n$ servers\footnote{%
    This parameter is denoted $s$ in \cite{rumyantsev_stability}.
}, where
jobs arrive according to a Poisson process with rate $\lambda$.
Each job requires $k$ servers with i.i.d. probability $p_k$,
and requires $\Exp(\mu)$ time to complete.
We will refer to a job requiring $k$ servers as a ``class $k$ job.''
Jobs are served FCFS, as in this work.

The state descriptor
consists of two parts:
The ``phase vector'' $m$ of the number of servers required by the
$n$ oldest jobs present in the system
(padded with empty entries if necessary),
and
the number of other jobs in the queue.
Let $m_i$ denote the number of servers
required by the $i$th oldest job in the queue.
The number of jobs in service
is defined to be $\sigma(m)$.
Let $\mathcal{M}$ denote the set of phase vectors
where $n$ jobs are present,
so no padding is needed.

With the model specified,
we may state the stability region result:
\begin{theorem}
    \label{thm:simplify}
    The (non-saturated) system is positive recurrent if
    \begin{align}
        \label{eq:rm_stability}
        \frac{\lambda}{\mu} \sum_{m \in \mathcal{M}} \frac{\prod_{j=1}^n p_{m_j}}{\sigma(m)}
        <1
    \end{align}
    and unstable if ``<'' is replaced by ``>''.
\end{theorem}
Note that we do not address the case of exact equality (which \cite{rumyantsev_stability} does).
We believe that our method
could be extended to cover this case,
but it would significantly complicate the proof.

Let us define the saturated system
to always contain exactly $n$ jobs,
so the state of the saturated system
corresponds to a phase vector in $\mathcal{M}$.
We will define the embedded Markov chain as usual,
transitioning on each departure.
The embedded Markov chain of the saturated system
has an extremely simple product-form steady-state distribution:
\begin{theorem}
    \label{thm:rm_embedded}
    The steady state distribution of the embedded DTMC of the saturated system is:
    \[ \pi_m = \prod_{j=1}^n p_{m_j}. \]
\end{theorem}
\begin{proof}
    Let us consider the balance equation for a given state $m$:
    \begin{align}
        \label{eq:rm_balance}
        \pi_m = \sum_{m' \in \mathcal{M}} \pi_{m'} P(m', m).
    \end{align}
    For each possible number of servers required $1 \le k \le n$,
    a state could transition to $m$ via the completion of a class $k$ job.
    Let $m^k(1)$ denote one specific predecessor state to $m$,
    the state $m^k(1) = [k, m_1, m_2, \ldots, m_{n-1}]$.
    Starting in state $m^k(1)$, if the class $k$ job in position 1 completes,
    and then a class $m_n$ job arrives, the system transitions to state $m$.

    Several more states can transition to state $m$
    via the completion of a class $k$ job.
    If a state has a class $k$ job in service,
    and the other $n-1$ jobs in the phase vector
    are in classes $m_1$ through $m_{n-1}$, in that order,
    then the state will transition to state $m$
    if the class $k$ job completes.
    In total, there are $\sigma(m^k(1))$ such states,
    because the class $k$ job will be in service
    if and only if it is in one of the positions
    that received service in state $m^k(1)$.\footnote{
        We will count otherwise identical states where the inserted job
        is inserted in distinct locations separately in this proof.
    }
    We will call these states $m^k(i)$, for $1 \le i \le \sigma(m^k(1))$:
    \begin{align*}
        m^k(1) &= [k, m_1, m_2, \ldots, m_{n-1}],
        m^k(2) = [m_1, k, m_2, \ldots, m_{n-1}],
        \ldots\\
        m^k(\sigma(m^k(1))) &= [m_1, m_2, \ldots, m_{\sigma(m^k(1))-1},
        k, m_{\sigma(m^k(1))}, \ldots, m_{n-1}]
    \end{align*}
    Note that $\sigma(m^k(i)) = \sigma(m^k(1))$
    for all $i \le \sigma(m^k(1))$,
    because the same jobs are served in each state.

    From state $m^k(i)$, the probability $P(m^k(i), m)$ of transitioning
    to state $m$ after a departure is
    $\frac{p_{m_n}}{\sigma(m^k(1))}.$
    This holds because each of the $\sigma(m^k(i)) = \sigma(m^k(1))$ jobs in service
    are equally likely to depart,
    because each job has the same exponential completion rate.
    If the $i$th oldest job departs, a class $m_n$ job must then arrive
    to complete the transition to state $m$.

    The steady state probability of state $m^k(i)$ under the steady state guess is
    $\pi_{m^k(i)} = \pi_m \frac{p_k}{p_{m_n}}$.
    Thus,
    \[ \sum_{i=1}^{\sigma(m^k(1))} \pi_{m^k(i)} P(m^k(i), m) = \pi_m p_k. \]
    Summing over all classes $k$,
    we see that the balance equation \cref{eq:rm_balance} must hold.
    Since $m$ was chosen arbitrarily,
    the balance equation holds for all states,
    and the steady state distribution is correct.
\end{proof}

Via the equivalent of \cref{thm:continuous_saturated},
because each saturated state $m$ has departure rate $\mu \sigma(m)$,
we find that the saturated system has throughput
\[ X = \left( \sum_{m \in \mathcal{M}} \frac{\prod_{j=1}^n p_{m_j}}{\mu \sigma(m)} \right)^{-1}. \]
Via the equivalent of \cref{thm:stability},
the non-saturated system is stable if $\lambda < X$
    and unstable if the opposite inequality holds,
which proves \cref{thm:simplify}.
Thus, we can now see that the simple form of the stability result 
in \citep{rumyantsev_stability}
emerges directly from the simple product-form steady state
of the  embedded Markov chain of the saturated system.

\section{Lessons learned about wastage}
\label{sec:lessons}

Our analytical results, proven in
\cref{thm:embedded_saturated,thm:continuous_saturated,thm:stability},
give us a closed-form expression for the stability region of the two-class multiserver-job system.
These results enable us to 
derive important insights into the behavior of such systems,
across a variety of parameter regimes.
All of the plots in this section are consequences of our analytical formula,
with no simulation needed.

\subsection{Number of servers wasted can approach $n_2$}
\label{sec:number_wasted}
Our first key insight is that the number of servers wasted in steady state can approach $n_2$.
In any given state,
at most $n_2-1$ servers can be wasted,
because if $n_2$ servers are free, any job can fit.
Our results show that under some conditions,
most of those $n_2$ servers can indeed be wasted in steady state,
yielding very high wastage and very poor utilization,
especially if $n_2$ is close to $n$.

\begin{figure}
    \begin{tabular}{cc}
        (a) \bm{$n_1 = 1, \, n_2 = 100, \, n=200$} &
        (b) \bm{$n_1 = 1, \, n_2 = 200, \, n=200$}\\
        \includegraphics[width=0.4\textwidth]{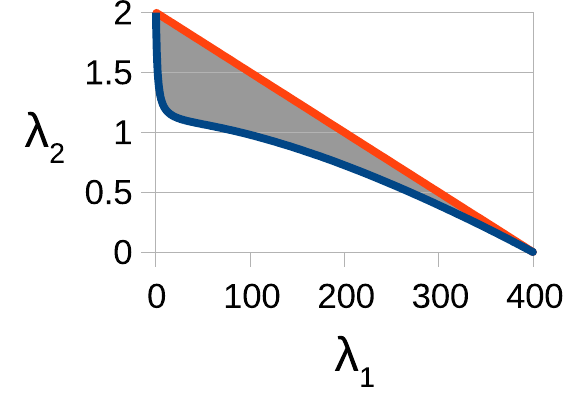} &
        \includegraphics[width=0.4\textwidth]{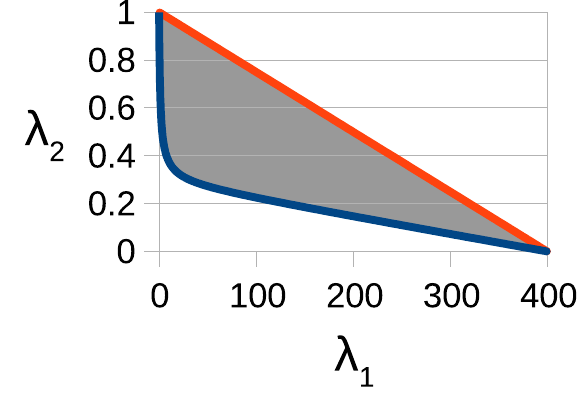}
    \end{tabular}
    \vspace{-16pt}
    \caption{The blue line depicts the stability region
    of the multiserver-job system, as the arrival rate of jobs in each class varies.
    For comparison, the red line
    shows where the stability region would lie if jobs could always be packed perfectly onto servers,
    so no servers were ever wasted.
    The grey shaded region shows the set of arrival rates where the system is unstable due to wastage.
    In each system, $\mu_1 = 2$ and $\mu_2 = 1$.}
    \label{fig:high_wastage}
\end{figure}

For example, consider the two systems depicted in \cref{fig:high_wastage}.
In both cases, the stability regions achieved are far from ideal.
The red lines illustrate the ``naive'' stability region,
i.e. where the stability region would lie
if jobs were always packed perfectly onto servers,
so no server was ever wasted.
The shaded region illustrates the reduction in the stability region due to wastage.
The reduction can be relatively moderate or more severe,
depending on the proportion of jobs which are in each class,
$p_1 = \lambda_1/\lambda$ and $p_2 = \lambda_2/\lambda$.
In \cref{fig:high_wastage}(a),
wastage reaches a peak when $p_2 \sim 6.4\%$,
meaning that $93\%$ of server-seconds are demanded by class 2 jobs
(since they require so many more servers).
In this case, $77$ servers are wasted in steady state, nearing $n_2 = 100$,
which results in a system utilization of only $61\%$.
\Cref{fig:high_wastage}(b) depicts an even more extreme situation, as $n_2 = n$.
Here, wastage reaches a peak when $p_2 \sim 1.3\%$,
meaning that $84\%$ of server-seconds are demanded by class 2 jobs.
In this case, $125$ servers are wasted in steady state, compared to $n_2 = 200$,
which results in a system utilization of only $37\%$.

In general, we observe that wastage is at its worst when
$n_2$ nears $n$,
when most jobs are class 1 jobs,
and when most server-seconds are demanded by class 2 jobs.
For wastage to be at its worst,
class 1 jobs must be common enough to consistently prevent
class 2 jobs from fully utilizing the servers,
but rare enough
to not heavily utilize the servers themselves.
Note that in either system in \cref{fig:high_wastage},
if all jobs come from just one class,
no wastage occurs.
Wastage only arises from the interleaving of class 1 and class 2 jobs,
with both classes of jobs blocking each other from service.

\subsection{Wastage falls as $\mu_2/\mu_1$ falls if $n_2$ divides $n$}
\label{sec:wastage_falls}

Our next insight is that wastage is highly dependent
on the relative service rates of the two classes.
Throughout this section we assume that $n_2$ divides $n$,
which is common in production systems
where the number of servers required by a job
and the total number of servers
are often powers of two.
If class 2 jobs require much more time
than class 1 jobs, i.e. if $\mu_2 \ll \mu_1$,
then few servers will be wasted.
If $\mu_1 \simeq \mu_2$ or if $\mu_1 \ll \mu_2$,
then more servers can be wasted, potentially approaching $n_2$.

This result partially rests on the fact that
if we hold constant all aspects of a system other than the service rates
(i.e. if we hold $n_1, n_2, n, p_1$ and $p_2$ constant),
then the number of wasted servers
is only a function of the service rate ratio $\mu_2/\mu_1$,
regardless of the absolute service rates $\mu_1$ and $\mu_2$.

\begin{claim}
    \label{clm:ratio}
    For a given ratio $\mu_2/\mu_1$,
    and for given values of $n_1, n_2, n, p_1$, and $p_2$,
    the long-term average number of wasted servers in the saturated system
    is not dependent on the specific values of $\mu_1$ and $\mu_2$.
\end{claim}
\begin{proof}
    Deferred to \cref{app:ratio}.
\end{proof}
To gain intuition into why wastage falls as $\mu_2/\mu_1$ falls
when $n_2$ divides $n$,
we will examine our analytical results in two asymptotic regimes.
As $\mu_2/\mu_1 \rightarrow 0$,
i.e. class 2 jobs require much more time than class 1 jobs,
the steady-state concentrates on the state with only class 2 jobs present,
namely $s_1(0)$.
In this state, $n\mod n_2$ servers are wasted,
which equals zero when class 2 jobs perfectly pack the available servers.
In general, (i.e. when $n_2$ does not perfectly divide $n$)
the steady state wastage will approach $n\mod n_2$.

As $\mu_2/\mu_1 \rightarrow \infty$,
the situation is much more complicated.
The steady state concentrates on the set of states with only class 1 jobs in service.
There are several such states, due to the possibility of a blocking job.
Looking at the stationary distribution, we see that
the most common such state is the state with the fewest class 1 jobs present,
while still having only class 1 jobs in service.
This state has $\lceil \frac{n-n_2+1}{n_1} \rceil$ class 1 jobs in service.
If $n_2$ is much larger than $n_1$,
the distribution of additional class $1$ jobs beyond the minimum is roughly a geometric distribution
with probability of continuation $p_1$.
Therefore, the number of wasted servers when $n_2$ is much larger than $n_1$
is approximately $n_2 - n_1/p_2$.
Therefore, we see that in the two extremes (when $n_2$ divides $n_1$),
wasted servers move from $n_2 - n_1/p_2$ as $\mu_2/\mu_1 \to \infty$
to zero, as $\mu_2/\mu_1 \to 0$.
We see this in \cref{fig:lessons_waste}.

\begin{figure}
\begin{tabular}{cc}
    (a) \bm{$n_1 = 1, \, n_2 = 10, \, n=30$} &
    (b) \bm{$n_1 = 1, \, n_2 = 100, \, n=200$}\\
    \includegraphics[width=0.4\textwidth]{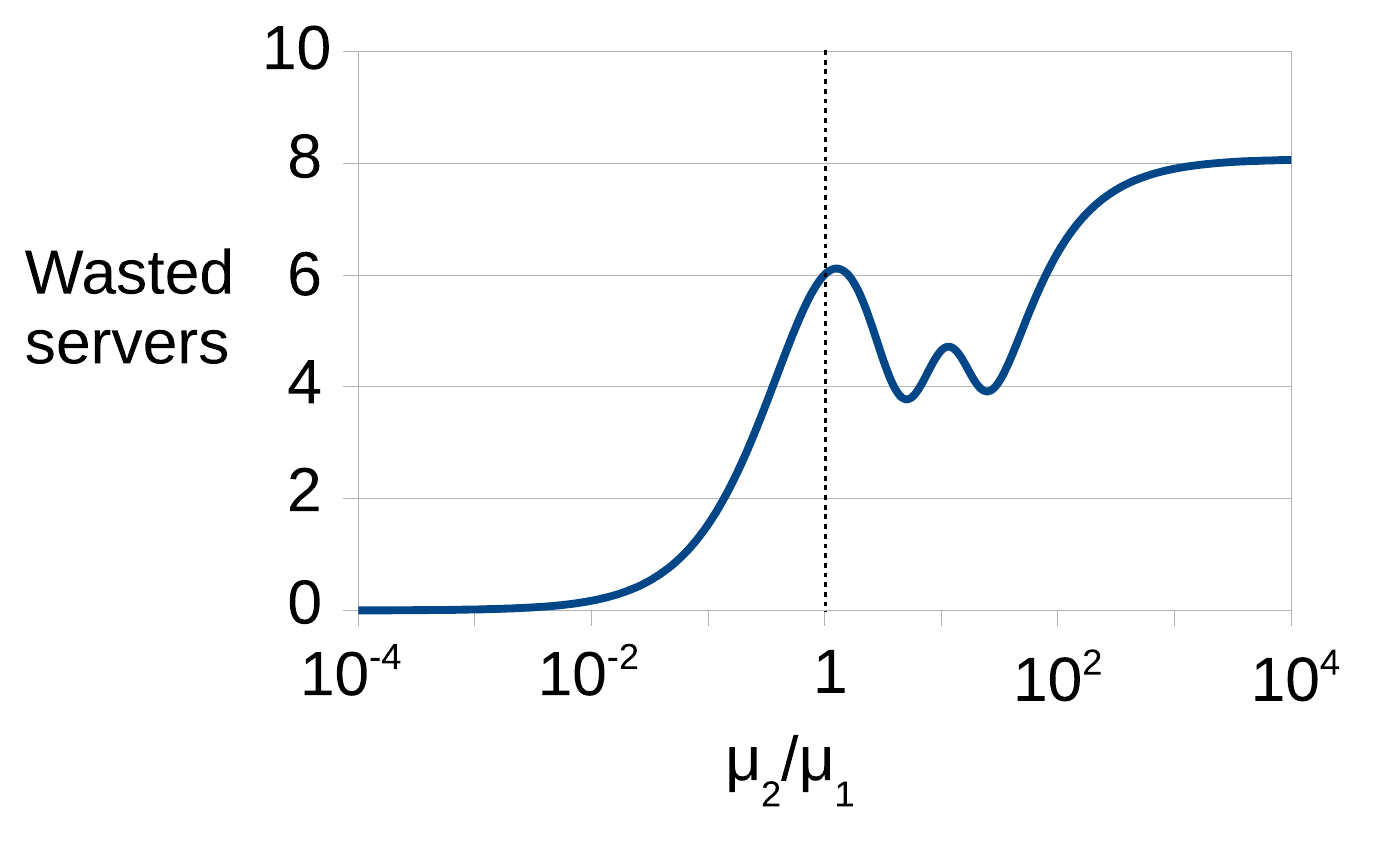} &
	\includegraphics[width=0.4\textwidth]{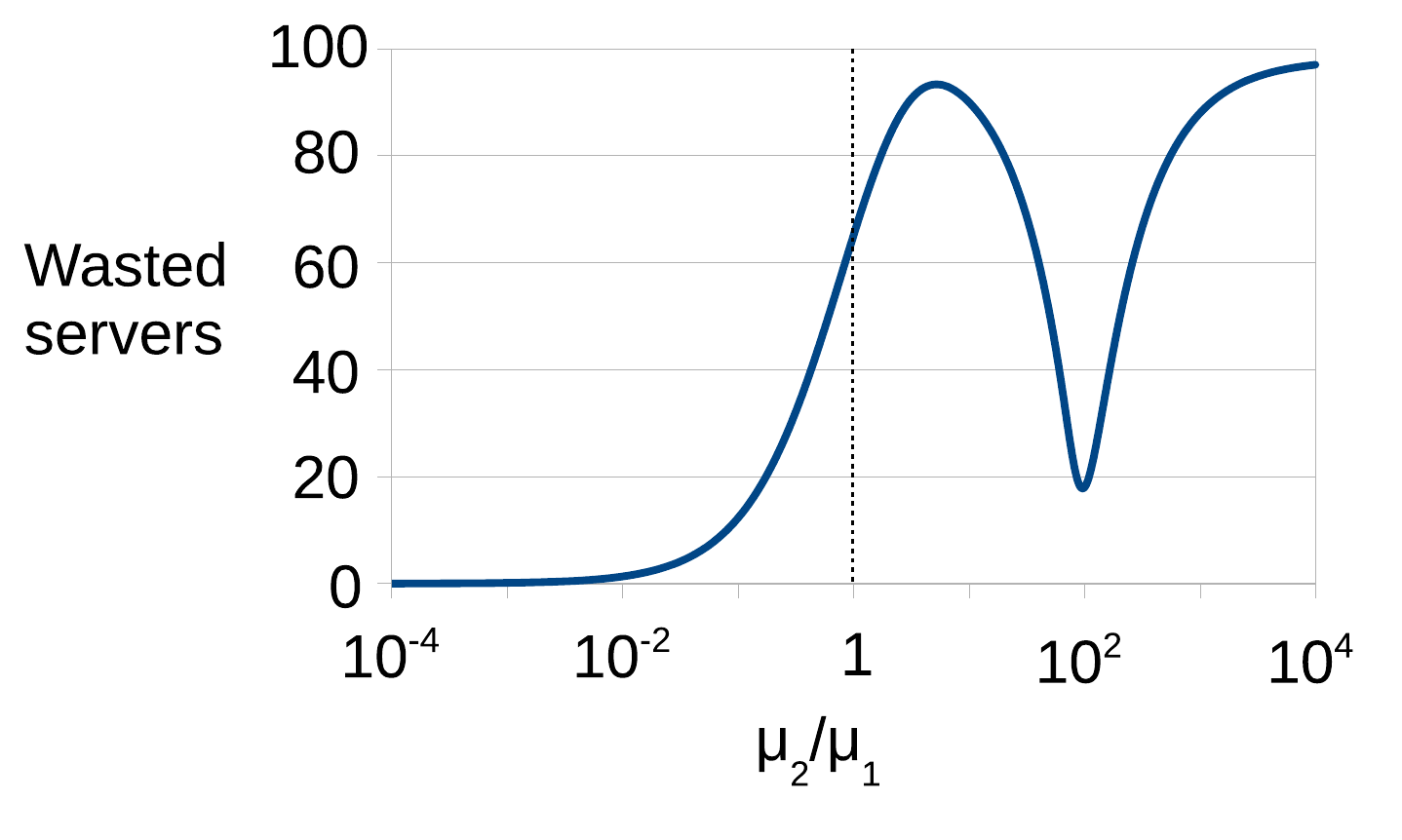}
\end{tabular}
    \vspace{-16pt}
    \caption{
Number of wasted servers in the saturated system in steady state
    as a function of the service rate ratio $\mu_2/\mu_1$,
in systems with $p_1 = p_2 = 0.5$.
    }
    \label{fig:lessons_waste}
\end{figure}

We also chart the number of wasted servers
as a function of the service rate ratio $\mu_2/\mu_1$
with fixed $p_1$ and $p_2$,
as shown in \cref{fig:lessons_waste}.
In both figures, few servers are wasted when $\mu_2/\mu_1$ is very small.
As $\mu_2/\mu_1$ increases to $1$, wastage grows to $60\%$ of $n_2$ in these cases.
As $\mu_2/\mu_1$ rises further,
many local optima and pessima exist,
but as the ratio gets very high, wastage stabilizes at almost all of $n_2$,
or more specifically $n_2 - n_1/p_2$.

\subsection{Wastage is nonmonotonic and idiosyncratic}
\label{sec:nonmonotonicity}
We also observe that our model can have nonmonotonic behavior,
with wastage rising and falling in idiosyncratic patterns,
rather than as part of larger trends.
Nonmonotonic behavior is pronounced when $n_2 \gg n_1$,
and muted or nonexistent when $n_1$ is closer to $n_2$.

\begin{figure}
\begin{tabular}{cc}
    (a) \bm{$n_1 = 1, \, n_2 = 67, \, n=201$} &
    (b) \bm{$n_1 = 3, \, n_2 = 10, \, n=30$}\\
    \includegraphics[width=0.4\textwidth]{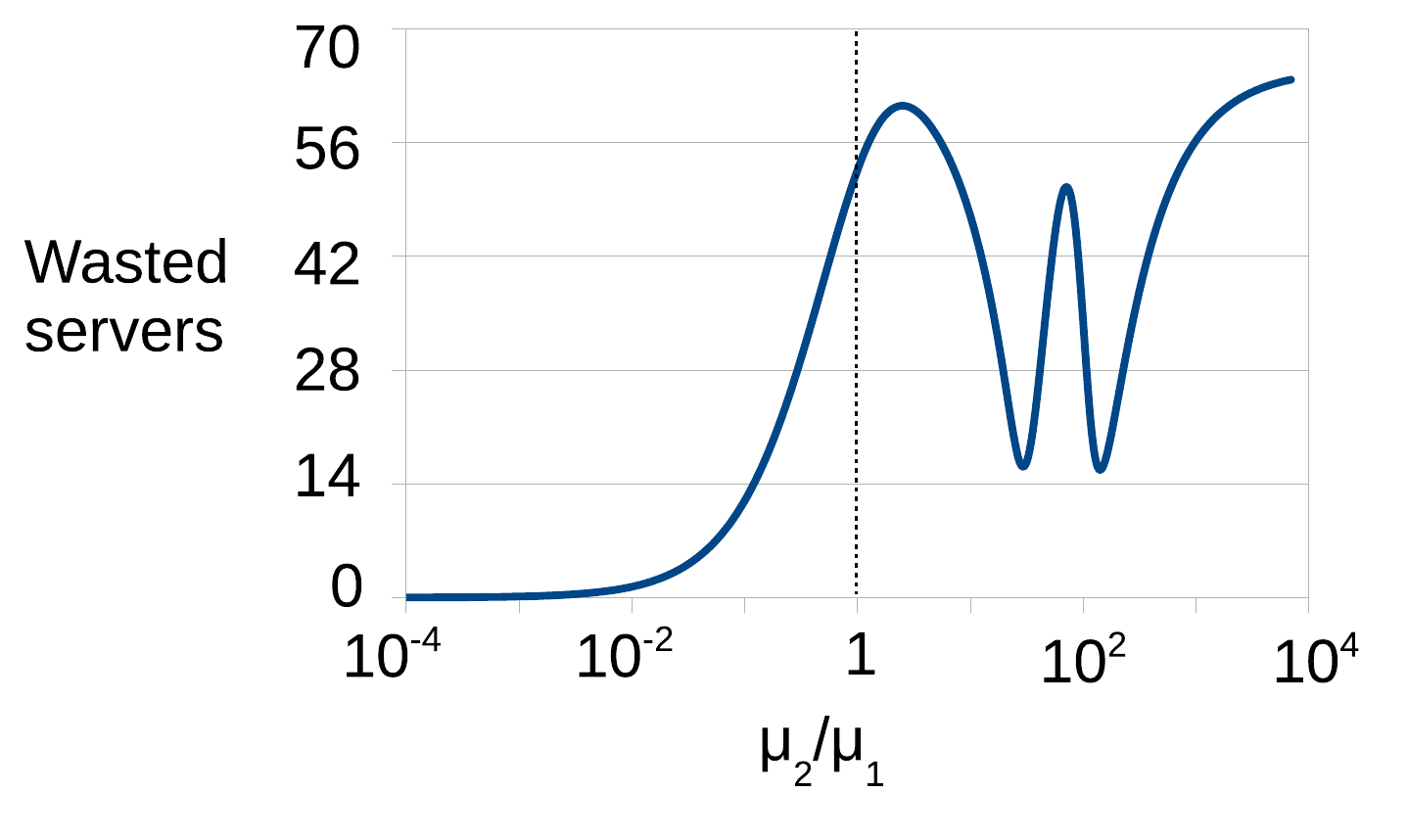} &
	\includegraphics[width=0.4\textwidth]{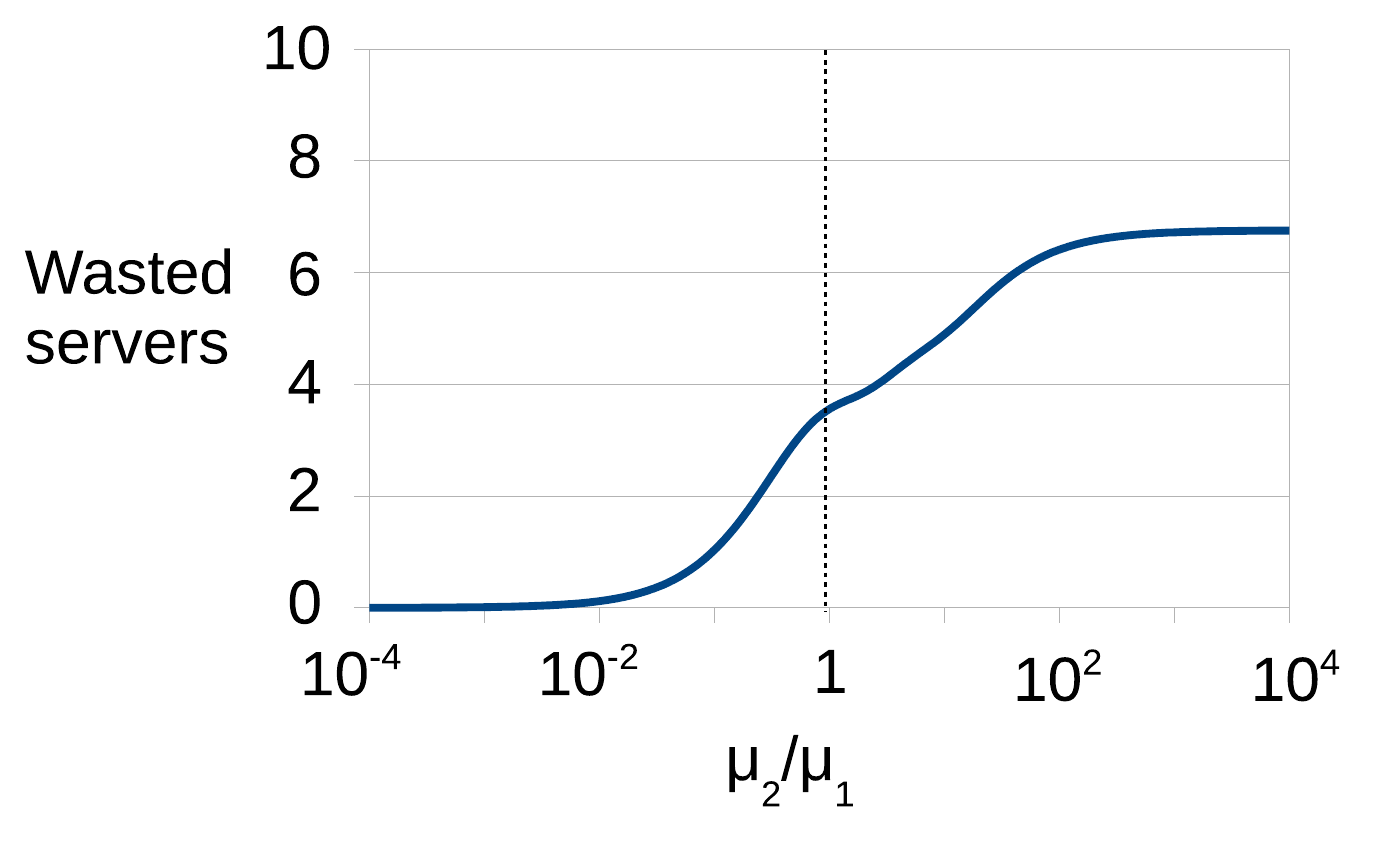}
\end{tabular}
    \vspace{-16pt}
    \caption{
Number of wasted servers in the saturated system in steady state
    as a function of the service rate ratio $\mu_2/\mu_1$,
in systems with $p_1 = p_2 = 0.5$.
    Note that in (a) where $n_2 \gg n_1$,
    wastage is highly nonmonotonic as a function of the service rate ratio.
    In contrast, in (b) where $n_2 \not\gg n_1$,
    wastage is monotonic as a function of the service rate ratio.%
}
    \label{fig:lessons_nonmonotonic}
\end{figure}

For an example of this pattern, compare \cref{fig:lessons_waste}(a)
and \cref{fig:lessons_nonmonotonic}(a).
In both cases, $n/n_2 = 3$, leading to a 3-peaked pattern
with the peaks occurring at similar values of $\mu_2/\mu_1$.
However, in \cref{fig:lessons_nonmonotonic}(a),
where $n_2$ is much larger than $n_1$,
the peaks and troughs are very pronounced,
while in \cref{fig:lessons_waste}(a) they are far more moderate.
In \cref{fig:lessons_nonmonotonic}(b),
$n_1$ and $n_2$ are closer still, and the number of wasted servers
is fully monotonic as a function of $\mu_2/\mu_1$.
We posit that
dips occur when class 1 jobs displace a class 2 job,
leading to significant wastage if $n_2 \gg n_1$.
However,
the number of class 1 jobs is not completely concentrated on a single value.
Rather, the number of class 1 jobs
varies around its mean by about $1/p_2$ jobs,
as discussed in \cref{sec:wastage_falls}.
These extra jobs seem to counterbalance up the changes in wastage that would
otherwise occur when a class 2 job is displaced from service.
This means that if $n_1/p_2$ is on the same order as $n_2$,
the dips tend to get completely covered up, making wastage fully monotonic.

In addition to nonmonotonic behavior displayed when $\mu_2/\mu_1$ is varied,
such behavior also occurs when $p_1$ and $p_2$ are varied, or when $n$ is varied,
to name just a few possibilities.
In general, nonmonotonic behavior is more the rule than the exception in the multiserver-job model.

This pattern of nonmonotonic behavior is very important to understand
when provisioning such a system or studying its workload:
the presence of nonmonotonic behavior means that there are local optima to seek out
and local pessima to be avoided,
not just long-term trends to follow.

\section{Conclusion and Discussion}
\label{sec:discussion}
In multiserver-job systems,
jobs don't always fit neatly into the servers,
and the resulting wastage of servers can be very high and hard to predict.
Because no prior analytical work in the FCFS multiserver-job model
was capable of handling the setting where $\mu_1$ and $\mu_2$ differ,
the important and complex effect of the service rate ratio $\mu_2/\mu_1$
on wastage was not previously understood.

We derive the first analytical formula for the stability region
of a two-class multiserver-job system.
This result allows us provide a detailed examination of the patterns of wastage,
without resorting to noisy and computationally intensive methods such as simulation.

We employ a saturated-system technique to derive our results.
We start by defining and studying the saturated system,
where additional jobs are always present in the queue.
We then derive a product-form steady-state distribution for the saturated system,
in \cref{thm:embedded_saturated,thm:continuous_saturated}.
We finish by proving that the stability region of the original model
equals the throughput of the saturated system, in \cref{thm:stability}.
We also show in \cref{thm:simplify} that we can use our saturated-system technique
to give a more intuitive proof
of the stability region of the model in \cite{rumyantsev_stability}.

Our results focus on the two-class multiserver-job system.
A natural next step is to consider systems with three or more classes with distinct rates.
Unfortunately, in such models the steady state distribution of the saturated system
no longer satisfies our product-form solution.
Moreover, the decomposed stationary equations
used in the proof of \cref{thm:embedded_saturated} (\cref{eq:decompose_1})
no longer hold.
One possible direction of study would be to
approximate a system with three or more classes
by a related system with two classes,
perhaps by preserving the number of servers required by the largest and smallest classes,
or required by the two most common classes.

\bibliographystyle{plainnat}
\bibliography{refs}

\begin{thebibliography}{30}
\providecommand{\natexlab}[1]{#1}
\providecommand{\url}[1]{\texttt{#1}}
\expandafter\ifx\csname urlstyle\endcsname\relax
  \providecommand{\doi}[1]{doi: #1}\else
  \providecommand{\doi}{doi: \begingroup \urlstyle{rm}\Url}\fi

\bibitem[Afanaseva et~al.(2019)Afanaseva, Bashtova, and
  Grishunina]{afanaseva_stability}
Larisa Afanaseva, Elena Bashtova, and Svetlana Grishunina.
\newblock Stability analysis of a multi-server model with simultaneous service
  and a regenerative input flow.
\newblock \emph{Methodology and Computing in Applied Probability}, pages 1--17,
  2019.

\bibitem[{Armstrong} et~al.(2010){Armstrong}, {Zhang}, {Katz}, {Wilde}, and
  {Foster}]{armstrong_scheduling}
T.~G. {Armstrong}, Z.~{Zhang}, D.~S. {Katz}, M.~{Wilde}, and I.~T. {Foster}.
\newblock Scheduling many-task workloads on supercomputers: Dealing with
  trailing tasks.
\newblock In \emph{2010 3rd Workshop on Many-Task Computing on Grids and
  Supercomputers}, pages 1--10, 2010.

\bibitem[Arthurs and Kaufman(1979)]{arthurs_sizing}
E.~Arthurs and J.~S. Kaufman.
\newblock Sizing a message store subject to blocking criteria.
\newblock In \emph{Proceedings of the third international symposium on
  modelling and performance evaluation of computer systems: Performance of
  computer systems}, pages 547--564, 1979.

\bibitem[Brill and Green(1984)]{brill_queues}
Percy~H. Brill and Linda Green.
\newblock Queues in which customers receive simultaneous service from a random
  number of servers: A system point approach.
\newblock \emph{Management Science}, 30\penalty0 (1):\penalty0 51--68, 1984.
\newblock \doi{10.1287/mnsc.30.1.51}.
\newblock URL \url{https://doi.org/10.1287/mnsc.30.1.51}.

\bibitem[Cao et~al.(2013)Cao, Ro, and Yin]{cao_comparison}
Yang Cao, CheulWoo Ro, and JianWei Yin.
\newblock Comparison of job scheduling policies in cloud computing.
\newblock In \emph{Future information communication technology and
  applications}, pages 81--87. Springer, 2013.

\bibitem[Etsion and Tsafrir(2005)]{etsion_short}
Yoav Etsion and Dan Tsafrir.
\newblock A short survey of commercial cluster batch schedulers.
\newblock \emph{School of Computer Science and Engineering, The Hebrew
  University of Jerusalem}, 44221:\penalty0 2005--13, 2005.

\bibitem[Feitelson and Rudolph(1996)]{feitelson_toward}
Dror~G. Feitelson and Larry Rudolph.
\newblock Toward convergence in job schedulers for parallel supercomputers.
\newblock In Dror~G. Feitelson and Larry Rudolph, editors, \emph{Job Scheduling
  Strategies for Parallel Processing}, pages 1--26, Berlin, Heidelberg, 1996.
  Springer Berlin Heidelberg.
\newblock ISBN 978-3-540-70710-3.

\bibitem[Feitelson et~al.(2004)Feitelson, Rudolph, and
  Schwiegelshohn]{feitelson_parallel}
Dror~G Feitelson, Larry Rudolph, and Uwe Schwiegelshohn.
\newblock Parallel job scheduling--a status report.
\newblock In \emph{Workshop on Job Scheduling Strategies for Parallel
  Processing}, pages 1--16. Springer, 2004.

\bibitem[Filippopoulos and Karatza(2007)]{fillippopoulos_mm2}
D.~Filippopoulos and H.~Karatza.
\newblock An m/m/2 parallel system model with pure space sharing among rigid
  jobs.
\newblock \emph{Mathematical and Computer Modelling}, 45\penalty0 (5):\penalty0
  491 -- 530, 2007.
\newblock ISSN 0895-7177.
\newblock \doi{https://doi.org/10.1016/j.mcm.2006.06.007}.
\newblock URL
  \url{http://www.sciencedirect.com/science/article/pii/S0895717706002627}.

\bibitem[Foster(1953)]{foster_stochastic}
F.~G. Foster.
\newblock On the stochastic matrices associated with certain queuing processes.
\newblock \emph{Ann. Math. Statist.}, 24\penalty0 (3):\penalty0 355--360, 09
  1953.
\newblock \doi{10.1214/aoms/1177728976}.
\newblock URL \url{https://doi.org/10.1214/aoms/1177728976}.

\bibitem[{Ghaderi}(2016)]{ghaderi_randomized}
J.~{Ghaderi}.
\newblock Randomized algorithms for scheduling vms in the cloud.
\newblock In \emph{IEEE INFOCOM 2016 - The 35th Annual IEEE International
  Conference on Computer Communications}, pages 1--9, 2016.

\bibitem[{Guo} et~al.(2018){Guo}, {Guan}, and {Ke}]{guo_optimal}
M.~{Guo}, Q.~{Guan}, and W.~{Ke}.
\newblock Optimal scheduling of vms in queueing cloud computing systems with a
  heterogeneous workload.
\newblock \emph{IEEE Access}, 6:\penalty0 15178--15191, 2018.

\bibitem[Institute(2020)]{msi_queues}
Minnesota~Supercomputing Institute.
\newblock Queues, 2020.
\newblock URL \url{https://www.msi.umn.edu/queues}.

\bibitem[Jones and Nitzberg(1999)]{jones_scheduling}
James~Patton Jones and Bill Nitzberg.
\newblock Scheduling for parallel supercomputing: A historical perspective of
  achievable utilization.
\newblock In Dror~G. Feitelson and Larry Rudolph, editors, \emph{Job Scheduling
  Strategies for Parallel Processing}, pages 1--16, Berlin, Heidelberg, 1999.
  Springer Berlin Heidelberg.
\newblock ISBN 978-3-540-47954-3.

\bibitem[Kim(1979)]{kim_mms}
Sung~Shick Kim.
\newblock \emph{M/M/s queueing system where customers demand multiple server
  use}.
\newblock PhD thesis, Southern Methodist University, 1979.

\bibitem[Kurowski et~al.(2013)Kurowski, Oleksiak, Piatek, and
  Weglarz]{kurowski_hierarchical}
Krzysztof Kurowski, Ariel Oleksiak, Wojciech Piatek, and Jan Weglarz.
\newblock Hierarchical scheduling strategies for parallel tasks and advance
  reservations in grids.
\newblock \emph{Journal of Scheduling}, 16\penalty0 (4):\penalty0 349--368,
  2013.

\bibitem[Madni et~al.(2017)Madni, Abd~Latiff, Abdullahi, Abdulhamid, and
  Usman]{madni_performance}
Syed Hamid~Hussain Madni, Muhammad~Shafie Abd~Latiff, Mohammed Abdullahi,
  Shafi'i~Muhammad Abdulhamid, and Mohammed~Joda Usman.
\newblock Performance comparison of heuristic algorithms for task scheduling in
  iaas cloud computing environment.
\newblock \emph{PLOS ONE}, 12\penalty0 (5):\penalty0 1--26, 05 2017.
\newblock \doi{10.1371/journal.pone.0176321}.
\newblock URL \url{https://doi.org/10.1371/journal.pone.0176321}.

\bibitem[{Maguluri} and {Srikant}(2014)]{maguluri_scheduling}
S.~T. {Maguluri} and R.~{Srikant}.
\newblock Scheduling jobs with unknown duration in clouds.
\newblock \emph{IEEE/ACM Transactions on Networking}, 22\penalty0 (6):\penalty0
  1938--1951, 2014.

\bibitem[{Maguluri} et~al.(2012){Maguluri}, {Srikant}, and
  {Ying}]{maguluri_stochastic}
S.~T. {Maguluri}, R.~{Srikant}, and L.~{Ying}.
\newblock Stochastic models of load balancing and scheduling in cloud computing
  clusters.
\newblock In \emph{2012 Proceedings IEEE INFOCOM}, pages 702--710, 2012.

\bibitem[Morozov and Rumyantsev(2016)]{morozov_stability}
Evsey Morozov and Alexander Rumyantsev.
\newblock Stability analysis of a map/m/s cluster model by matrix-analytic
  method.
\newblock In Dieter Fiems, Marco Paolieri, and Agapios~N. Platis, editors,
  \emph{Computer Performance Engineering}, pages 63--76, Cham, 2016. Springer
  International Publishing.
\newblock ISBN 978-3-319-46433-6.

\bibitem[Psychas and Ghaderi(2017)]{psychas_nonpreemptive}
Konstantinos Psychas and Javad Ghaderi.
\newblock On non-preemptive vm scheduling in the cloud.
\newblock \emph{Proc. ACM Meas. Anal. Comput. Syst.}, 1\penalty0 (2), December
  2017.
\newblock \doi{10.1145/3154493}.
\newblock URL \url{https://doi.org/10.1145/3154493}.

\bibitem[Rumyantsev and Morozov(2017)]{rumyantsev_stability}
Alexander Rumyantsev and Evsey Morozov.
\newblock Stability criterion of a multiserver model with simultaneous service.
\newblock \emph{Annals of Operations Research}, 252\penalty0 (1):\penalty0
  29--39, 2017.

\bibitem[Sliwko(2019)]{sliwko_taxonomy}
Leszek Sliwko.
\newblock A taxonomy of schedulers--operating systems, clusters and big data
  frameworks.
\newblock \emph{Global Journal of Computer Science and Technology}, 2019.

\bibitem[{Tang} et~al.(2011){Tang}, {Lan}, {Desai}, {Buettner}, and
  {Yu}]{tang_reducing}
W.~{Tang}, Z.~{Lan}, N.~{Desai}, D.~{Buettner}, and Y.~{Yu}.
\newblock Reducing fragmentation on torus-connected supercomputers.
\newblock In \emph{2011 IEEE International Parallel Distributed Processing
  Symposium}, pages 828--839, 2011.

\bibitem[{Tang} et~al.(2012){Tang}, {Ren}, {Lan}, and {Desai}]{tang_adaptive}
W.~{Tang}, D.~{Ren}, Z.~{Lan}, and N.~{Desai}.
\newblock Adaptive metric-aware job scheduling for production supercomputers.
\newblock In \emph{2012 41st International Conference on Parallel Processing
  Workshops}, pages 107--115, 2012.

\bibitem[Tikhonenko(2005)]{tikhonenko_generalized}
Oleg~M Tikhonenko.
\newblock Generalized erlang problem for service systems with finite total
  capacity.
\newblock \emph{Problems of Information Transmission}, 41\penalty0
  (3):\penalty0 243--253, 2005.

\bibitem[Tirmazi et~al.(2020)Tirmazi, Barker, Deng, Haque, Qin, Hand,
  Harchol-Balter, and Wilkes]{tirmazi_borg}
Muhammad Tirmazi, Adam Barker, Nan Deng, Md~E. Haque, Zhijing~Gene Qin, Steven
  Hand, Mor Harchol-Balter, and John Wilkes.
\newblock Borg: The next generation.
\newblock In \emph{Proceedings of the Fifteenth European Conference on Computer
  Systems}, EuroSys '20, New York, NY, USA, 2020. Association for Computing
  Machinery.
\newblock ISBN 9781450368827.
\newblock \doi{10.1145/3342195.3387517}.
\newblock URL \url{https://doi.org/10.1145/3342195.3387517}.

\bibitem[van Dijk(1989)]{vandijk_blocking}
Nico~M. van Dijk.
\newblock Blocking of finite source inputs which require simultaneous servers
  with general think and holding times.
\newblock \emph{Operations Research Letters}, 8\penalty0 (1):\penalty0 45 --
  52, 1989.
\newblock ISSN 0167-6377.
\newblock \doi{https://doi.org/10.1016/0167-6377(89)90033-3}.
\newblock URL
  \url{http://www.sciencedirect.com/science/article/pii/0167637789900333}.

\bibitem[Vizino et~al.(2005)Vizino, Kochmar, Stone, and Scott]{vizino_batch}
Chad Vizino, J~Kochmar, N~Stone, and R~Scott.
\newblock Batch scheduling on the cray xt3.
\newblock \emph{CUG 2005}, 2005.

\bibitem[Whitt(1985)]{whitt_blocking}
Ward Whitt.
\newblock Blocking when service is required from several facilities
  simultaneously.
\newblock \emph{AT\&T technical journal}, 64\penalty0 (8):\penalty0 1807--1856,
  1985.

\end{thebibliography}
\appendix
    \section{Proofs deferred from Section~\ref{sec:embedded}}
    Here we present the proofs of various cases and lemmas
    that were deferred from \cref{sec:embedded}.
    \subsection{Class 1 completions, blocking job}
    \label{sec:class_1_blocking}
    Let the current state be $[1, a, b]$.
    We want to show that the class 1 decomposed stationary equation must hold:
    \[ p_1 \pi_{[1, a, b]} = \sum_{s'} \pi_{s'} P_1(s', [1, a, b]). \]
    To start with, we enumerate the states that can transition to $[1, a, b]$
    via the completion of a class 1 job.

    By inspecting \cref{lem:possible_transitions},
    we can see that transitions to blocking states via class 1 completions
    only occur in cases $ii$ and $vi$,
    corresponding to potential predecessors
    $[0, a+1, b]$ and $[1, a+1, b]$, respectively.
    Only one of these states exists, namely state $s_1(a+1)$.
    Thus, the only possible state that can transition to $[1, a, b]$ via the completion
    of a class 1 job is the state $s_1(a+1)$.

    If state $s_1(a+1)$ is $[1, a+1, b]$,
    the completion of the class 1 job is not followed by any arrivals.
    This transition has probability $f_1([1, a+1, b])$.
    If state $s_1(a+1)$ is $[0, a+1, b]$, the completion of the
    class 1 job is followed by a class 2 job becoming the blocking job. 
    This transition has probability $f_1([0, a+1, b])p_2$.
    For instance, in the 3-10-30 system, if the current state is $[1, 4, 1]$,
    then the prior state could only be $[1, 5, 1]$;
    if the current state is $[1, 2, 2]$,
    then the prior state could only be $[0, 3, 2]$.

    We will prove that the decomposed stationary equation holds both when $s_1(a+1)$
    is $[0, a+1, b]$, and when $s_1(a+1)$ is $[1, a+1, b]$.

    If $s_1(a+1)$ is $[0, a+1, b]$, then we must show that
    \begin{align*}
        p_1 \pi_{[1, a, b]} &= \pi_{[0, a+1, b]} f_1([0, a+1, b]) p_2
        = p_2 \pi_{[0, a+1, b]} f_1(s_1(a+1)).
    \end{align*}
    This follows immediately from the steady state guess,
    by comparing the expressions for states $[1, a, b]$ and $[0, a+1, b]$.

    If $s_1(a+1)$ is $[1, a+1, b]$, then we must show that
    \begin{align*}
        p_1 \pi_{[1, a, b]} &= \pi_{[1, a+1, b]} f_1([1, a+1, b])
        = \pi_{[1, a+1, b]} f_1(s_1(a+1)).
    \end{align*}
    This follows immediately from the steady state guess,
    by comparing the expressions for states $[1, a, b]$ and $[1, a+1, b]$.

    With both scenarios covered, the decomposed stationary equation must hold in this case.

    \subsection{Class 2 completions, blocking job}
    \label{sec:class_2_blocking}
    Let the current state be $[1, a, b]$.
    We want to show that the class 2 decomposed stationary equation must hold:
    \[ p_2 \pi_{[1, a, b]} = \sum_{s'} \pi_{s'} P_2(s', [1, a, b]) \]
    To start with, we enumerate the states that can transition to $[1, a, b]$
    via the completion of a class 2 job.
    
    By inspecting \cref{lem:possible_transitions},
    we can see that transitions to blocking states via class 2 completions
    only occur in cases $v$ and $viii$.
    Case $v$ corresponds to the predecessor state $s_2(b+1)$.
    Let us write $s_2(b+1)$ as $[0, a-i^*, b+1]$.
    Then case $viii$ corresponds to the set of possible predecessor states
    $[1, a-i, b]$ where $0 \le i < i^*$.
    As a result,
    these states are all the possible states that could transition to $[1, a, b]$
    after a class 2 job completes.

    The state $s_2(b+1)$ could transition to $[1, a, b]$
    if a class 2 job completed,
    then a series of class 1 jobs arrived until $a$ were in service,
    then a class 2 job arrived to block the head of the queue.
    This transition has probability $f_2([0, a-i^*, b+1]) p_1^{i^*} p_2$.
    For instance, in the 3-10-30 system,
    if the current state is $[1, 5, 1]$,
    then the possible prior state is $s_2(2) = [0, 3, 2]$, and $i^*$ is 2.

    Another set of possible prior states are the states $[1, a-i, b]$
    for any $i$ such that $0 \le i < i^*$,
    including the state $[1, a, b]$ itself.
    These states could transition to $[1, a, b]$ if a class 2 job completed,
    the class 2 job blocking the head of the queue entered service,
    then $i$ more class 1 jobs entered service from the queue,
    and finally a class 2 job arrived to block the head of the queue once more.
    This transition has probability $f_2([1, a-i, b]) p_1^i p_2$.
    For instance, in the 3-10-30 system,
    if the current state is $[1, 5, 1]$,
    then the possible prior states in this category are $[1, 4, 1]$ and $[1, 5, 1]$.

    To prove the decomposed stationary equation, we must show that
    \begin{align*}
        p_2 \pi_{[1, a, b]} &= \pi_{[0, a-i^*, b+1]} f_2([0, a-i^*, b+1]) p_1^{i^*} p_2
        + \sum_{i=0}^{i^*-1} \pi_{[1, a-i, b]}f_2([1, a-i, b]) p_1^i p_2\\
        \iff \pi_{[1, a, b]} &= \pi_{[0, a-i^*, b+1]} f_2([0, a-i^*, b+1]) p_1^{i^*}
                + \sum_{i=0}^{i^*-1} \pi_{[1, a-i, b]}f_2([1, a-i, b]) p_1^i
    \end{align*}
    Applying \cref{lem:telescope} with $q=0$ and $r=i^*-1$,
    our desired statement simplifies to
    \begin{align*}
        \pi_{[1, a, b]} = p_1^{i^*} \pi_{[0, a-i^*, b+1]} f_2([0, a-i^*, b+1])
        &+ \pi_{[1, a, b]} - p_1^{i^*-1} \pi_{[1, a-i^*+1, b]} f_1([1, a-i^*+1,b])\\
        \iff p_1^{i^*} \pi_{[0, a-i^*, b+1]} f_2([0, a-i^*, b+1])
        &= p_1^{i^*-1} \pi_{[1, a-i^*+1, b]} f_1([1, a-i^*+1,b])\\
        \iff p_1 \pi_{[0, a-i^*, b+1]} f_2([0, a-i^*, b+1])
        &= \pi_{[1, a-i^*+1, b]} f_1([1, a-i^*+1,b])
    \end{align*}
    This follows immediately from the steady state guess, by comparing the expressions
    for states $[0, a-i^*, b+1]$ and $[1, a-i^*+1, b]$.

    Thus, the decomposed stationary equation holds in this case.
    \subsection{Edge cases}
    \label{sec:edge_cases}
    In \cref{sec:class_1_no_blocking,sec:class_2_no_blocking,sec:class_1_blocking,sec:class_2_blocking},
    we verified that the decomposed stationary equations held,
    but in doing so we assumed that certain states existed.
    In particular, in \cref{sec:class_1_no_blocking,sec:class_1_blocking}
    (for class 1 completions)
    we assumed that state $s_1(a+1)$ existed,
    and in \cref{sec:class_2_no_blocking,sec:class_2_blocking}
    (for class 2 completions)
    we assumed that state $s_2(b+1)$ existed.
    In certain states, these neighboring states do not exist,
    so we verify that the decomposed stationary equations still hold here.

    We will split up this section to handle distinct kinds of edge cases:
    states where $s_1(a+1)$ does not exist,
    and states where $s_2(b+1)$ does not exist.
    We will further subdivide the latter case
    by whether any class 1 jobs are present in the system.

    Note that a state where $s_1(a+1)$ does not exist
    is a state containing the maximum possible number of class 1 jobs,
    and similarly for states where $s_2(b+1)$ does not exist.

    \subsubsection{Maximum possible number of class 1 jobs}

    If the state $s_1(a+1)$ does not exist,
    this means that $a+1$ class 1 jobs cannot fit in the server.
    Therefore, $a$ class 1 jobs must completely fill the server.
    In particular, this means that $s_1(a)$ is $[0, a, 0]$.
    For instance, in the 3-10-30 system,
    this edge case occurs for state $[0, 10, 0]$.
    There is no state $s_1(11)$.

    We must show that the class 1 decomposed stationary equation
    holds for the state $[0, a, 0]$.
    
    For a state to transition to $[0, a, 0]$ on a class 1 completion,
    there must be no class 2 jobs present in the prior state.
    The only state with no class 2 jobs present is $[0, a, 0]$ itself.
    This transition has probability $f_1([0, a, 0]) p_1$.

    To verify the class 1 decomposed stationary equation in this edge case,
    we must show that
    \[ p_1 \pi_{[0, a, 0]} = \pi_{[0, a, 0]} p_1 f_1([0, a, 0]). \]
    Note that in state $[0, a, 0]$, a class 1 job is guaranteed to complete next.
    In other words, $f_1([0, a, 0]) = 1$.
    Therefore, the decomposed stationary equation holds.

    \subsubsection{Maximum possible number of class 2 jobs}
    Next, we consider the case where the state $s_2(b+1)$ does not exist.
    This means that $b+1$ class 2 jobs cannot fit in the server.

    In a given system, there can be one or more states in this edge case.
    For instance, in the 3-10-30 system, this edge case occurs for state $[0, 0, 3]$.
    There is no state $s_2(4)$.

    For another example, in a different system with $n_1 = 1$, $n_2 = 4$, $n=7$,
    this edge case occurs for states $[0, 3, 1]$, $[1, 2, 1]$, $[1, 1, 1]$,
    and $[1, 0, 1]$. There is no state $s_2(2)$ in that system.

    Note that there always exists a state with no class 1 jobs present
    among the states in this edge case.
    We will handle the state without any class 1 jobs separately from
    the other states in this case.

    \subsubsection{No class 1 jobs}
    
    Call the state with no class 1 jobs $[h, 0, b]$.
    For a state to transition to $[h, 0, b]$ on a class 2 completion,
    there must be no class 1 jobs present in the prior state.
    The only state with no class 1 jobs present is $[h, 0, b]$
    itself.
    This transition has probability $f_2([h, 0, b])p_2$.
    To verify the class 2 decomposed stationary equation in this edge case,
    we must show that
    \[ p_2 \pi_{[h, 0, b]} = \pi_{[h, 0, b]} f_2([h, 0, b]) p_2 \]
    Note that in state $[h, 0, b]$,
    a class 2 job is guaranteed to complete next.
    In other words, $f_2([h, 0, b]) = 1$.
    Therefore, the decomposed stationary equation holds.

    \subsubsection{At least one class 1 job}
    Next, we consider the case where state $s_2(b+1)$ does not exist,
    but there is at least one class 1 job in the system.

    In the system with $n_1 = 1, n_2 = 4, n =7$,
    these are the states $[0, 3, 1]$, $[1, 2, 1]$, and $[1, 1, 1]$.

    We will first consider states in this edge case that do not have a blocking job.
    Let the state be $[0, a, b]$.
    For a state to transition to $[0, a, b]$ on a class 2 completion,
    the prior state must have at most $a$ class 1 jobs.
    This means it must have at least $b$ class 2 jobs.
    Since $s_2(b+1)$ does not exist, the prior state must have exactly $b$ class 2 jobs.

    The prior state can be $[0, a, b]$ itself,
    if a class 2 job completes, and then a class 2 job arrives.
    This transition happens with probability $f_2([0, a, b]) p_2$.
    The prior state can also be a state of the form $[1, a-i, b]$,
    where $1 \le i \le a$.
    This transition can happen if a class 2 job completes,
    then $i$ class 1 jobs arrive.
    This transition happens with probability $f_2([1, a-i, b]) p_1^i$.
    Since $a \ge 1$, there is at least one such state.

    To verify the class 2 decomposed stationary equation in this case,
    we must show that
    \begin{align*}
        p_2 \pi_{[0, a, b]} &= \pi_{[0, a, b]} f_2([0, a, b]) p_2
        + \sum_{i=1}^a \pi_{[1, a-i, b]} f_2([1, a-i, b]) p_1^i \\
        \iff p_2 \pi_{[0, a, b]} f_1([0, a, b])
        &= \sum_{i=1}^a \pi_{[1, a-i, b]} f_2([1, a-i, b]) p_1^i
    \end{align*}
    Applying \cref{lem:telescope} with $q=1$ and $r=a$,
    our desired statement simplifies to
    \begin{align*}
        p_2 \pi_{[0, a, b]} f_1([0, a, b]) =
        p_1 \pi_{[1, a-1, b]} - p_1^a \pi_{[1, 0, b]} f_1([1, 0, b])
    \end{align*}
    Since $[1, 0, b]$ has no class 1 jobs present,
    $f_1([1, 0, b]) = 0$.
    Thus, our desired statement simplifies to
    \begin{align*}
        p_2 \pi_{[0, a, b]} f_1([0, a, b]) = p_1 \pi_{[1, a-1, b]}
    \end{align*}
    This follows immediately from the steady state guess,
    by comparing the expressions for states $[0, a, b]$ and $[1, a-1, b]$.

    Finally, we consider the case where state $s_2(b+1)$ does not exist,
    there is at least one class 1 job in the system, and there is a blocking job
    at the head of the queue.
    Let the state be $[1, a, b]$.
    For a state to transition to $[1, a, b]$ on a class 2 completion,
    the prior state must have at most $a$ class 1 jobs.
    This means it must have at least $b$ class 2 jobs.
    Since $s_2(b+1)$ does not exist, the prior state must have exactly $b$ class 2 jobs.

    The prior state can be any state of the form $[1, a-i, b]$,
    where $0 \le i \le a$.
    This transition can happen if a class 2 job completes,
    the blocking job enters service,
    $i$ class 1 jobs arrive,
    and then a class 2 job arrives to block the servers again.
    This transition happens with probability $f_2([1, a-i, b]) p_1^i p_2$.

    To verify the class 2 decomposed stationary equation in this case,
    we must show that
    \begin{align*}
        p_2 \pi_{[1, a, b]} &= \sum_{i=0}^a \pi_{[1, a-i, b]} f_2([1, a-i, b]) p_1^i p_2\\
        \iff \pi_{[1, a, b]} &= \sum_{i=0}^a \pi_{[1, a-i, b]} f_2([1, a-i, b]) p_1^i
    \end{align*}
    Applying \cref{lem:telescope} with $q=0$ and $r=a$,
    our desired statement simplifies to
    \begin{align*}
        \pi_{[1, a, b]} = \pi_{[1, a, b]} - p_1^a \pi_{[1, 0, b]} f_1([1, 0, b])
    \end{align*}
    Since $f_1([1, 0, b]) = 0$, the statement must hold.

    We have verified that the decomposed stationary equations hold in all cases.
\subsection{Proof of Lemma~\ref{lem:telescope}} 
\label{sec:proof_telescope}
    \begin{replemma}{lem:telescope}
        For all $a, b$ and all $q \le r$ such that both $[1, a-q, b]$
        and $[1, a-r, b]$ are valid states,
        under the steady state guess in \cref{eq:steady_state_guess},
        \[
            \sum_{i=q}^r p_1^i \pi_{[1, a-i, b]} f_2([1, a-i, b])
            = p_1^q \pi_{[1, a-q, b]} - p_1^r \pi_{[1, a-r, b]}f_1([1, a-r, b])
        \]
    \end{replemma}
    \begin{proof}
        We shall proceed by induction on $r$, for any given values of $a, b,$ and $q$.

        First, note that the base case $r=q$ merely states that
        \[p_1^q \pi_{[1, a-q, b]} f_2([1, a-q, b])
        = p_1^q \pi_{[1, a-q, b]} - p_1^q \pi_{[1, a-q, b]} f_1([1, a-q, b])
        \]
        This is true from the definitions of $f_1$ and $f_2$ in \cref{eq:def_f}.

        For the inductive case, let us assume that
        \begin{align}
            \label{eq:step_jm1}
            \sum_{i=q}^{r-1} p_1^i \pi_{[1, a-i, b]} f_2([1, a-i, b])
            = p_1^q \pi_{[1, a-q, b]} - p_1^{r-1} \pi_{[1, a-(r-1), b]}f_1([1, a-(r-1), b])
        \end{align}
        Next, note that by comparing the steady state guess for the states $[1, a-r, b]$
        and $[1, a-(r-1), b] = [1, a-r+1, b]$,
        we find that
        \begin{align*}
            \pi_{[1, a-r, b]} p_1 &= \pi_{[1, a-r+1, b]} f_1(s_1(a-r+1)) 
            = \pi_{[1, a-r+1, b]} f_1([1, a-r+1, b])
        \end{align*}
        Performing this substitution into \cref{eq:step_jm1},
        we find that
        \begin{align*}
            \sum_{i=q}^{r-1} p_1^i \pi_{[1, a-i, b]} f_2([1, a-i, b])
            = p_1^q \pi_{[1, a-q, b]} - p_1^r \pi_{[1, a-r, b]}
        \end{align*}
        Next, we add $p_1^r \pi_{[1, a-r, b]}  f_2([1, a-r, b])$ to both sides of the equation,
        giving
        \begin{align*}
            \sum_{i=q}^r p_1^i \pi_{[1, a-i, b]} f_2([1, a-i, b])
            = p_1^q \pi_{[1, a-q, b]} - p_1^r \pi_{[1, a-r, b]} f_1([1, a-r, b])
        \end{align*}
        Thus, the inductive case holds.
    \end{proof}

\section{Proof of Theorem~\ref{thm:continuous_saturated}}
\label{sec:saturated_proof}
\begin{reptheorem}{thm:continuous_saturated}
    The steady state distribution of the saturated system is:
    \[ \p_{[h, a, b]} = X \frac{\pi_{[h, a, b]}}{a \mu_1 + b \mu_2}\]
    where $X$, the normalizing constant,
    is the throughput of the saturated system:
    \[ X = \left( \sum_{[h, a, b]} \frac{\pi_{[h, a, b]}}{a \mu_1 + b \mu_2} \right)^{-1}.\]
\end{reptheorem}
\begin{proof}
    From \cref{thm:embedded_saturated},
    we know the steady state probability $\pi_{s}$
    that the embedded DTMC of the saturated system is in a given state $s$.
    In particular, we know that the distribution $\pi$ satisfies the
    balance equations for the embedded DTMC for all states $s$:
    \begin{align}
        \label{eq:old_steady_state}
        \pi_s = \sum_{s'} \pi_{s'} P(s', s)
    \end{align}
    where $P(s, s')$ denotes the probability of a transition from
    $s$ to $s'$.

    To show that a distribution $\p_s$ is the steady state
    of the continuous-time saturated system,
    we must show that for all states $s$,
    \begin{align}
        \label{eq:steady_state}
        \p_s \nu_s = \sum_{s'} \p_{s'} \nu_{s'} P(s', s),
    \end{align}
    where $\nu_s$ denotes the rate of transitions out of state $s$.
    In particular, let $\p$ denote the distribution
    \[
        \p_s = C' \frac{\pi_s}{\nu_s}
    \]
    for a normalization constant $C'$.
    Then, due to \cref{eq:old_steady_state},
    one can easily see that $\p$ satisfies \cref{eq:steady_state},
    and so $\p$ is the steady state distribution
    of the continuous-time saturated system.

    Note that in a given state $[h, a, b]$
    with $a$ class 1 jobs in service and $b$ class 2 jobs in service,
    the rate of transitions out of $[h, a, b]$
    is equal to the completion rate:
    \[
        \nu_{[h, a, b]} = a \mu_1 + b \mu_2
    \]
    This relationship holds because a transition occurs on each completion
    in the saturated system.
    Recall that we consider self-transitions to be transitions.

    Thus, the steady state distribution $\p$ is
    \[
        \p_{[h, a, b]} = C' \frac{\pi_{[h, a, b]}}{a \mu_1 + b \mu_2} \text{ where }
        C' = \left( \sum_{[h,a,b]} \frac{\pi_{[h, a, b]}}{a \mu_1 + b \mu_2} \right)^{-1},
    \]
    as desired.

    To better understand the value of $C'$,
    note that the time between departures in the continuous-time saturated system is
    distributed as $\Exp(a \mu_1 + b \mu_2)$ with probability $\pi_{[h, a, b]}$.
    Thus, the mean time between departures is
    \[
        E[\text{time between departures}] = \sum_{[h,a,b]} \frac{\pi_{[h, a, b]}}{a \mu_1 + b \mu_2}.
    \]
    The system throughput $X$ must be equal to the reciprocal of the mean interdeparture time,
    matching $C'$, as desired.

\end{proof}

\section{Proof of Lemma~\ref{lem:foster}}
\label{app:foster}
    \begin{replemma}{lem:foster}
        Let $V$ be Lyapunov function which maps each state in the
        embedded Markov chain of the Augmented Saturated System
        to the value of its jobs counter.
        Then $V$ satisfies the conditions of Foster's theorem,
        showing that the embedded chain is positive recurrent.
    \end{replemma}
    \begin{proof}

    Foster's theorem considers an irreducible Markov chain
    on a countable state space $S$ with transition probabilities $P_{ij}$
    for $i, j \in S$.
    Foster's theorem states that the Markov chain is positive recurrent if
    there exists a Lyapunov function $V : S \to \mathbb{R}$
    and a finite set $F$ such that
    \begin{enumerate}
    \item $V(i) \ge 0$ for all $i \in S$,
    \item $\sum_{j \in S} P_{ij} V(j) < \infty$ for all $i \in F$, and
    \item $\sum_{j \in S} P_{ij} V(j) \le V(i) - \epsilon$ for all $i \not\in F$,
    for some $\epsilon > 0$.
    \end{enumerate}

    We will use the value of the jobs counter as the Lyapunov function $V$.
    The value of the jobs counter is always non-negative, satisfying the first property.
    As for the other properties,
    note that the change in the value of the jobs counter over some interval
    is equal to the number of arrivals during the interval
    minus the number of completions during the interval,
    except when the value of the jobs counter has already reached zero
    and further completions occur.
    Let $C_0(t, i)$ denote the number of completions that occur while
    the jobs counter has reached zero over an interval of length $t$,
    starting from state $i$.
    Let $J(t, i)$ denote the change in the value of the jobs function
    over an interval of length $t$, starting from state $i$.
    Then we can symbolically state that
    \[ J(t, i) = A(t, i) - (C(t, i) - C_0(t, i)) \]
    We wish to bound $E[J(t_1, i)]$,
    in order to show that the requirements of Foster's theorem are satisfied.
    We know that
    \[ E[J(t_1, i)] = \sum_{j \in S} P_{ij} (V(j) - V(i)). \]

    First, let us consider the second requirement of Foster's theorem.
    Regardless of the choice of finite set $F$,
    there must be some finite upper bound on the value of the jobs counter
    for jobs in $F$.
    Then, after one step, the expected value of the jobs counter
    increases by at most
    \[E[J(t_1, i)] \le E[A(t_1, i)] = \lambda t_1.\]
    The result must be finite,
    so the second requirement of Foster's theorem is satisfied,
    regardless of the choice of $F$.
    
    As for the third requirement,
    note that if the initial value of the jobs counter is very large,
    specifically much larger than $\mu_1 t_1$ and $\mu_2 t_1$,
    there is a very low probability that the
    jobs counter will reach zero during the next $t_1$ time,
    and $E[C_0(t_1, i)]$ over that interval will be very small.
    In particular, there must exist some $c$ such that for any state $i'$
    with jobs counter at least $c$, the expected number of jobs
    completed while the counter is zero
    is at most $0.1$.
    For any such state $i'$, the expected change in the value of the jobs counter can be bounded:
    \begin{align*}
    &E[A(t_1, i')] - E[C(t_1, i')] + E[C_0(t_1, i')]\\
    &\le \lambda t_1 - (\lambda t_1 + 1) + 0.1 = -0.9
    \end{align*}
    Therefore, let us define $F$ to be the finite set of states with jobs counter
    less than $c$.
    By the above argument,
    the third requirement of Foster's theorem is satisfied with $\epsilon = 0.9$.
    
    With all three requirements satisfied,
    Foster's theorem tells us that the embedded Markov Chain
    of the Augmented Saturated System must be positive recurrent.
\end{proof}

\section{Proof of Claim~\ref{clm:ratio}} 
\label{app:ratio}
\begin{repclaim}{clm:ratio}
    For a given ratio $\mu_2/\mu_1$,
    and for given values of $n_1, n_2,$ and $n$,
    the long-term average number of wasted servers in the saturated system
    is not dependent on the specific values of $\mu_1$ and $\mu_2$.
\end{repclaim}
\begin{proof}
    We prove this claim in two ways:
    using \cref{thm:embedded_saturated,thm:continuous_saturated},
    and directly via a time-scaling argument.

    For the first proof,
    let us write out the steady-state expected number of wasted servers.
    First, note that in a given state $[h, a, b]$,
    the number of wasted servers is $n - a n_1 - b n_2$.
    Therefore, using the result of \cref{thm:continuous_saturated},
    the expected number of wasted servers in steady-state is
    \[ E[\text{number of wasted servers}]
    = X \sum_{[h, a, b]} \frac{\pi_{[h, a, b]} (n - a n_1 - b n_2)}{a \mu_1 + b \mu_2} \]
    where $X$ is the throughput:
    \begin{align}
        \label{eq:throughput}
        X  = \left( \sum_{[h,a,b]} \frac{\pi_{[h,a,b]}}{a \mu_1 + b \mu_2} \right)^{-1}
    \end{align}

    There are four components to this formula:
    the steady state of the embedded saturated system, $\pi_{[h, a, b]}$,
    the number of wasted servers, $n - a n_1 - b n_2$,
    the completion rate, $a \mu_1 + b \mu_2$,
    and the throughput, $X$.

    First, we will show that $\pi_{[h, a, b]}$ is only dependent on the ratio $\mu_2/\mu_1$,
    not the specific values of $\mu_1$ and $\mu_2$.
    Examining \cref{eq:steady_state_guess} in \cref{thm:embedded_saturated},
    we can see that $\pi_{[h, a, b]}$ depends on $p_1$ and $p_2$,
    as well as $f_1(\cdot)$ and $f_2(\cdot)$,
    the probabilities that the next job to complete from a given state
    is either a class 1 or class 2 job, respectively.
    Note that by the definition of $f_1(\cdot)$ and $f_2(\cdot)$
    in \cref{eq:def_f},
    $f_1(\cdot)$ and $f_2(\cdot)$ are only dependent on the ratio
    $\mu_2/\mu_1$, not the specific values of $\mu_1$ and $\mu_2$.
    As a result, the same is true of $\pi_{[h, a, b]}$.

    Next, consider the throughput $X$, given by \cref{eq:throughput}.
    Because $\pi_{[h, a, b]}$ is only dependent on the ratio $\mu_2/\mu_1$,
    if we scale $\mu_1$ and $\mu_2$ by the same constant,
    $X$ is scaled by that constant as well.
    As a result, the ratio
    \[ \frac{X}{a\mu_1 + b \mu_2} \]
    is only dependent on the ratio $\mu_2/\mu_1$
    for any given $a$ and $b$.

    Combining everything together,
    we find that each term of the form
    \[ X \frac{\pi_{[h, a, b]} (n - a n_1 - b n_2)}{a \mu_1 + b \mu_2} \]
    is only dependent on the ratio $\mu_2/\mu_1$.
    Thus, the expected number of wasted servers in steady-state
    is likewise only dependent on the ratio $\mu_2/\mu_1$, as desired.

    For the second proof, note that increasing $\mu_1$ and $\mu_2$
    by some multiplicative factor $c$
    is equivalent to speeding up the passage of time by a factor of $c$,
    because $\mu_1$ and $\mu_2$ are the only time-dependent quantities
    in the saturated system.
    Speeding up the passage of time by a factor of $c$ cannot change
    the average number of wasted servers in the saturated system,
    because the same steady state will be observed.
\end{proof}
\end{document}